\documentclass[aoas, noinfoline]{imsart}
\usepackage[ansinew]{inputenc}
\usepackage{graphicx} 
\usepackage{textcomp} 
\usepackage{amsfonts} 
\usepackage{amsmath}
\usepackage{amssymb}
\usepackage{latexsym}
\usepackage{array}
\usepackage{float}
\usepackage{amsthm} 
\usepackage{calc}
\usepackage{perpage}
\usepackage{textcomp}
\usepackage{verbatim} 
\usepackage{marvosym}
\usepackage{hieroglf}
\usepackage{multirow}
 \usepackage{rotating}
\usepackage{dsfont} %for nice indicator
\usepackage{tikz}
\usepackage{url}

\usepackage[english]{babel}
\usepackage{hyperref} 
\usepackage{color}

\newtheorem{theorem}{Theorem}[section]

\newtheorem{lemma}{Lemma}[section]

\newtheorem{remark}{Remark}[section]

\newcommand{\1}{{\rm 1}\mskip -4,5mu{\rm l} }

\newcommand{\argmin}{\mathop{\mathrm{arg\,min}}}

\newcommand{\op}[1]{\operatorname{#1}}  % Lettre droite!!!!

\newcommand{\sign}{\op{sign}}

\def\R{\mathbb{R}}

\newcommand{\rp}{\mathbb{R}^p}
\newcommand{\rn}{\mathbb{R}^n}

\newcommand{\mpr}{\mathbb{P}}

\newcommand{\indexlasso}{\{1,\dots,p\}}
\newcommand{\est}{\widehat\beta}
\newcommand{\set}{{\widehat S}}
\newcommand{\setnor}{{|\widehat S|}}
\newcommand{\estset}{{\widehat\beta_{\widehat S}}}
\newcommand{\matset}{X_{\widehat S}}
\newcommand{\graset}{X_{\widehat S}^TX_{\widehat S}}
\newcommand{\grasetinv}{(X_{\widehat S}^TX_{\widehat S})^{\text -1}}
\newcommand{\rset}{\R^{\setnor}}
\newcommand{\betatrue}{\beta^*}

\newcommand{\estref}{\overline{\beta}}
\newcommand{\setref}{\overline S}
\newcommand{\estsetref}{{\estref_{\widehat S}}}

\newcommand{\estcri}{\widetilde{\beta}_c}

\pdfoutput=1
\usepackage{pdfsync}

\def\paper{1} % 1 for paper, 0 for review

\RequirePackage[OT1]{fontenc}
\RequirePackage[numbers,sort]{natbib}
\RequirePackage[colorlinks,citecolor=blue,urlcolor=blue]{hyperref}
\RequirePackage{hypernat}
\setattribute{journal}{name}{}

\begin{document}

\begin{frontmatter}
\title{Trust, but verify: benefits and pitfalls of ~least-squares~refitting~in high~dimensions}
\runtitle{Least-squares refitting in high dimensions}

\begin{aug}

\if\paper1 % Only for papers
\author{\fnms{Johannes Lederer\thanksref{t2}}\ead[label=e1]{johanneslederer@mail.de}}%\and
%\author{\fnms{Author 2}\thanksref{t2}\ead[label=e2]{author2@email2.com}}

\thankstext{t2}{The author acknowledges financial support from the Swiss National Science Foundation.}

\address[a]{~\\
Johannes Lederer\\
%$ \mathrm{{^c\!/\!_o}} $ Stephanie Gonye\\
%1876 Arch Street, Unit 1\\
%Berkeley, CA 94709\\
\printead{e1}\\
Homepage: \url{http://www.johanneslederer.de}    
}
%\address[b]{Some Address\\
%of author 2\\
%\printead{e2}}

\runauthor{Johannes Lederer}

\affiliation{University of California, Berkeley}
\fi

\end{aug}

\if\paper1 % Only for papers

\begin{abstract}~
Least-squares refitting is widely used  in high dimensional regression to reduce the prediction bias of $\ell_1$-penalized estimators (e.g., Lasso and Square-Root Lasso). We present theoretical and numerical results that provide new insights into the benefits and pitfalls of least-squares refitting. In particular, we consider both prediction and estimation, and we pay close attention to the effects of correlations in the design matrices of linear regression models, since these correlations - although often neglected - are crucial in the context of linear regression, especially in high dimensional contexts. First, we demonstrate that the benefit of least-squares refitting strongly depends on the setting and task under consideration: least-squares refitting can be beneficial even for settings with highly correlated design matrices but is not advisable for all settings, and least-squares refitting can be beneficial for estimation but performs better for prediction. Finally, we introduce a criterion that indicates whether least-squares refitting is advisable for a specific setting and task under consideration, and we conduct a thorough simulation study involving the Lasso to show the usefulness of this criterion.
\end{abstract}

\fi

\end{frontmatter}

\section{Introduction}
High dimensional statistical models enjoy increasing popularity in many fields of research. Particularly popular are regression models of the form 
\begin{equation}
  \label{eq:model}
  Y= X\betatrue +\sigma \epsilon
\end{equation}
with outcome $Y\in\rn$, design matrix $X\in\R^{n\times p}$, regression vector $\betatrue\in\rp$, and noise vector $\epsilon\in\rn$ with associated noise level $\sigma>0$. High dimensional applications of such regression models involve a number of parameters $p$ that is comparable to the number of observations $n$ or even much larger. Nevertheless, many of these applications also involve a sparsity level~$s$, which is the number of nonzero entries of the regression vector, that is considerably smaller than~$n$ and~$p$. In this context, $\ell_1$-penalized methods have exhibited excellent numerical and theoretical properties for the estimation of the regression vector~$\betatrue$ ({parameter estimation}), of the active set $S:=\{j\in\indexlasso:\betatrue_j\neq 0\}$ ({variable selection}), and of $X\betatrue$ ({prediction}) from the outcome $Y$ and the design matrix $X$. The corresponding initial estimators are
\begin{equation}
  \label{eq:estimator}
  \est := \argmin_{\beta\in\rp}\left\{g\left(\|Y-X\beta\|_2^2\right)+\lambda\|\beta\|_1\right\}
\end{equation}
for given real-valued function $g$ on $[0,\infty)$ and tuning parameter $\lambda>0$. Prominent examples are the Lasso~\cite{Tibshirani-LASSO} (see also~\cite{Buhlmann11} and references therein)
\begin{equation}
  \label{eq:lasso}
  \argmin_{\beta\in\rp}\left\{\|Y-X\beta\|_2^2+\lambda\|\beta\|_1\right\},
\end{equation}
which corresponds to $g: x \mapsto x$, and the Square-Root Lasso \cite{Belloni11,ScaledLasso11} (see also~\cite{Yoyo13})
\begin{equation}
  \label{eq:sqrtlasso}
  \argmin_{\beta\in\rp}\left\{\|Y-X\beta\|_2+\lambda\|\beta\|_1\right\},
\end{equation}
which corresponds to $g: x \mapsto \sqrt x$.\\

The estimates resulting from~\eqref{eq:estimator} are often treated further to circumvent some of their well-known shortcomings. To reduce the number of superfluously estimated parameters, for example, the entries of the estimates are often thresholded or included in the penalty term of a subsequent estimation, see~Remark~\ref{rm:othermethods} in Section~\ref{sec:theory} and~\cite{Candes08, Meinsthresh09, SaraShu11}. In this paper, however, we study least-squares refitting on the estimated active set $\set:=\{j\in\indexlasso:\est_j\neq 0\}$. This method was initially designed for removing the prediction biases but is in this paper considered for both prediction and estimation. The corresponding estimators, called LS refitted estimators in the following, are
\begin{equation}
  \label{eq:estimatorrefitted}
  \estref_\set:=\argmin_{\xi\in\rset}\|Y-X_{\set}\xi\|_2^2,~~~~\estref_{\set^c}:=0,
\end{equation}
where the subscripts indicate that the vectors and matrices are restricted to the entries and columns, respectively, with indices in the corresponding sets. It turns out that the level of correlation in the design matrix, that is, the magnitudes of the off-diagional entries of the matrix $X^TX$, play an important role in this context. If  - in the sense of restricted eigenvalues, see~\cite{Sara09} and references therein - the level of correlation  in the design matrix is low, the LS refitted estimators typically outperform the corresponding initial estimators regarding prediction; this follows from~\cite{Belloni09}, since for weakly correlated design matrices, the initial estimators typically provide good estimates of the active sets, and the restricted eigenvalue conditions are satisfied by construction.  For arbitrary design matrices, in contrast, a thorough comparison of the initial estimators and the corresponding LS~refitted estimators has not been made. Moreover, least-squares refitting can be applied for estimation as well but has not been studied thoroughly for this task, yet. Therefore, we have
\begin{itemize}
\item[\bf Goal 1:] Study least-square refitting for arbitrary design matrices.
\end{itemize}
This is of great interest, since
\begin{itemize}
\item[$-$] correlations are common, especially in high dimensional applications;
\item[$-$] for many applications, it is unclear whether the results for weakly correlated design matrices apply, since the corresponding restricted eigenvalues depend on the active set $S$ and can therefore not be calculated;
\item[$-$] even for highly correlated design matrices, the initial estimators can perform well, especially for prediction, and the LS refitted estimators could therefore be of interest for this case as well.
\end{itemize}
\begin{itemize}
\item[\bf Goal 2:] Study least-square refitting for estimation.
\end{itemize}
This is of great interest, since
\begin{itemize}
\item[$-$] estimation is needed in many applications.\\
\end{itemize}

In this paper, we relate the prediction and estimation errors of the LS~refitted estimators to those of the corresponding initial estimators; these relations hold for arbitrary design matrices and therefore supplement the ones in~\cite{Belloni09}, which cover prediction and weakly correlated design matrices only. We then complement the theoretical results with a thorough numerical comparison of the Lasso and its least-squares refitted version, which are the most popular estimators in the framework considered in this paper. First, we find 
\begin{itemize}
\item[\bf Result 1:] For both prediction and estimation, least-squares refitting can be beneficial even if the design matrix is highly correlated; However, least-squares refitting can be disadvantageous if the level of correlation is high and, at the same time, the sparsity level is considerably larger than 1.
\end{itemize}
\begin{itemize}
\item[\bf Result 2:] Least-squares refitting can be advisable for estimation but exhibits better performances for prediction.
\end{itemize}
These results provide a new insight into the properties of least-squares refitting but depend on the active set, which is not known beforehand in practice. Therefore, a crucial question remains: given an application, is least-squares refitting beneficial or should the initial estimates not be modified for this specific application? We address this question introducing adaptive estimators induced by a criterion that is designed to distinguish between favorable and unfavorable settings for least-squares refitting. We then test for various settings the numerical performance of the adaptive estimator that corresponds to the Lasso and the associated least-squares refitted estimator. We find for prediction and estimation
\begin{itemize}
\item[\bf Result 3:] The introduced adaptive estimators can outperform the Lasso and its least-squares refitted version.\\
\end{itemize}

The structure of the paper is as follows: First, in Section~\ref{sec:conventions}, we introduce the notation and some conventions comprising mild assumptions on the function $g$ and the noise vector $\epsilon$. We then turn to the main part of the paper, Section~\ref{sec:main}. Motivated by some illustrative simulation results given in Section~\ref{sec:fo}, we study theoretical and practical aspects of the LS~refitted estimators in Sections~\ref{sec:theory} and~\ref{sec:criteria}. In particular, we present in Section~\ref{sec:theory} bounds for the errors of the LS~refitted estimators that hold for arbitrary design matrices and introduce and numerically test in Section~\ref{sec:criteria} a criterion for the application of LS~refitted estimators. We then summarize our findings in Section~\ref{sec:discussion}. Finally, we give detailed proofs and further remarks in Section~\ref{sec:proofs}.

% Festlegung: Level of correlation. singular\\
% sparsity level.\\
% regression vector. no true\\
% Y = outcome\\
% Active set\\
% noise, distribution of the noise\\
% nonzero entries\\
% initial estimator\\
% ~\\

\subsection{Notation and conventions}\label{sec:conventions} %notation nicht notations
We assume that the function $g$ is convex and has a strictly positive derivative $g'$. This implies, in particular, that $\beta\mapsto g\left(\|Y-X\beta\|_2^2\right)$ is convex and a solution of~\eqref{eq:estimator} exists, see Lemma~\ref{res:convex} in Section~\ref{sec:techrm}. If the solution is not unique,  $\est$ is defined as one of the solutions with a minimal number of nonzero entries. Note also that the differentiability in $0$ is assumed only for ease of exposition and can be easily circumvented to include, for example, the Square-Root Lasso~\eqref{eq:sqrtlasso}, see Remark~\ref{rm:sqrtlassoextension} in Section~\ref{sec:techrm}. Next, the
design matrix $X$ is assumed to be nonrandom and normalized such that $\left(X^TX\right)_{jj}=n$ for $j\in\indexlasso$. Moreover, the noise is assumed to be nonsingular in the sense that the probability of the event $Y\in U$ is zero for any nonrandom subspace $U\subsetneq\rn$.\\

For ease of exposition, we introduce some additional, convenient notation. To this end, let $A\subset\indexlasso$ be a nonempty set, $d\in\{1,2,\dots,p\}$ an integer, and $v\in\R^d,~\beta\in\rp$ vectors. First, the cardinality of $A$ is denoted by $|A|$. Then, the vector consisting of the entries of $\beta$ with indices in $A$ is denoted by $\beta_A\in\R^{|A|}$, and the matrix consisting of the columns of $X$ with indices in $A$ is denoted by $X_A\in\R^{n\times |A|}$. Next, the indices of the nonzero entries of $v$ are denoted by $S(v):=\{j\in\{1,\dots,d\}: v_j\neq 0\}$, and $\sign(v):=(\sign(v_1),\dots,\sign(v_d))^T\in\R^d$ is the vector containing the signs of the entries of $v$. Finally, we set $\argmin_{\xi\in\R^{\setnor}}\|Y-X_\set\xi\|_2:=0$ if $\set=\emptyset$. However, since the results can be easily derived for $\set=\emptyset$, we only consider $\set\neq\emptyset$.

\section{Main results}\label{sec:main}

\subsection{First observations}\label{sec:fo} In this section, we numerically compare the Lasso with the corresponding LS refitted estimator. We consider three settings that differ from each other especially in the levels of correlation in the design matrices and the sparsity levels. The results indicate that the benefit of least-squares refitting crucially depends on the setting and the task under consideration.
 
\paragraph{Setting} We generated settings with different sparsity levels $s$ and levels of correlation in the design matrices~$X$. To this end, we first generated the columns of the design matrices according to
\begin{equation*}
  X_j:=\sqrt n~\frac{\kappa v+(1-\kappa)\xi_j}{\|\kappa v+(1-\kappa)\xi_j\|_2} \in \rn,
\end{equation*}
where the vectors $v,\xi_1,\dots,\xi_p$ were independently sampled from the standard normal distribution in $\rn$. The level of correlation is determined by the constant $\kappa\in[0,1]$: the larger $\kappa$, the more are the  columns of the design matrix correlated. We then set the regression vector to 
\begin{equation*}
\betatrue:=(1/s,2/s,\dots,1,0,\dots,0)^T\in\rp  
\end{equation*}
with sparsity level $|S|=s$. We finally generated the outcome $Y$ according to the model~\eqref{eq:model} fixing the standard deviation of the noise $\sigma$ and sampling the vector $\epsilon$  from the standard normal distribution in $\rn$.

\paragraph{Accessible quantities in real data applications} 
The outcome $Y$ and the design matrix $X$ (and therefore the sample size $n$ and the number of parameters~$p$) are known and yield the estimators $\est$ and $\estref$ via~\eqref{eq:estimator} and~\eqref{eq:estimatorrefitted}, respectively.  In contrast, the regression vector $\betatrue$ and especially the sparsity level $s$ are subject to estimation and not accessible. Next, the correlations in the design matrix can - in principle - be derived from the design matrix. However, the level of correlation is often measured with restricted eigenvalues or related quantities; these quantities depend on the sparsity level and are thus inaccessible.  Finally, the standard deviation of the noise $\sigma$ may be known or unknown. While many estimators include $\sigma$ or an estimate of it, there are also estimators that are adaptive with respect to $\sigma$ (for example, the Square-Root Lasso~\eqref{eq:sqrtlasso}). %The estimation of $\sigma$ is not in the scope of this paper.

\begin{table}
\begin{tabular}{l c c c c}
 & \bf{pred. error} & \bf est. error & \bf false neg. & \bf false pos.\\ 
\multicolumn{5}{c}{~}\\
\multicolumn{5}{c}{$n=1000,~p=1000,~\sigma=0.3,~s=2,~\kappa=0$}\vspace{0.2cm}\\
{\bf Lasso} & $(2.87 \pm 0.04)\times 10^{\text -3}$ & $(5.28\pm 0.03)\times 10^{\text -2}$ & $0.00\pm 0.00$ & $0.11\pm 0.01$  \\ 
{\bf LS Lasso}& $(0.34 \pm 0.02)\times 10^{\text -3}$ & $(1.50\pm 0.04)\times 10^{\text -2}$ & $0.00\pm 0.00$ & $0.11\pm 0.01$\\
{\bf Zero}&$1.25\pm 0.01$&$1.12$&2&0\\
~&~&~&~&~\\
\multicolumn{5}{c}{~}\\
\multicolumn{5}{c}{$n=100,~p=1000,~\sigma=0.3,~s=2,~\kappa=0.9$}\vspace{0.2cm}\\
{\bf Lasso} & $(2.31 \pm 0.03)\times 10^{\text -2}$ & $(9.38\pm 0.07)\times 10^{\text -1}$ & $0.87\pm 0.02$ & $11.6\pm 0.2$  \\ 
{\bf LS Lasso}& $(1.06 \pm 0.02)\times 10^{\text -2}$ & $(9.22\pm 0.07)\times 10^{\text -1}$ & $0.87\pm 0.02$ & $11.6\pm 0.2$ \\
{\bf Zero}&$2.24\pm 0.01$&$1.12$&2&0\\
~&~&~&~&~\\ 
\multicolumn{5}{c}{~}\\
\multicolumn{5}{c}{$n=100,~p=1000,~\sigma=0.3,~s=20,~\kappa=0.9$}\vspace{0.2cm}\\
{\bf Lasso} & $(5.10 \pm 0.04)\times 10^{\text -2}$ & $2.19\pm 0.02$ & $8.29\pm 0.08$ & $50.9\pm 0.4$\\ 
{\bf LS Lasso}& $(5.67 \pm 0.04)\times 10^{\text -2}$ & $3.51\pm 0.04$ & $8.29\pm 0.08$ & $50.9\pm 0.4$\\
{\bf Zero}&$108.99\pm 0.01$&$2.68$&20&0\\
~&~&~&~&~\\
\end{tabular}
\caption{Prediction, estimation, and variable selection performances of the Lasso~\eqref{eq:lasso}, the LS~Lasso (which is the corresponding LS refitted estimator~\eqref{eq:estimatorrefitted}), and the Zero estimator $0\in\rp$. The setting and computation are detailed in the section. For a plot of the relative prediction and estimation errors, see Figure~\ref{fig:foillu} below.}\label{tab:fo}
\end{table}

\paragraph{Computation} We considered the Lasso~\eqref{eq:lasso} and the corresponding LS refitted estimator \eqref{eq:estimatorrefitted}, called LS Lasso in the following, and compared their performances via their prediction and estimation errors and their variable selection properties. Additionally, we determined the performances of the Zero estimator 
\begin{equation}
  \label{eq:zero}
  {\beta}^{\text{zero}}:=0\in\rp
\end{equation}
to check the usefulness of the Lasso and the LS Lasso. More precisely, for any of the three estimators $\beta\in\{\est,\estref, \beta^{\text{zero}}\}$, the prediction error was set to $\|X\beta-X\betatrue\|_2^2/n$, the estimation error to $\|\beta-\betatrue\|_2$, and the variable selection properties were determined via the number of false negatives $|\{j\in S, j\notin S(\beta)\}|$ and false positives $|\{j\notin S, j\in S(\beta) \}|$. For all settings, we performed 1000 repetitions and denoted the corresponding means by 
\begin{align*}
  \text{pred. error}&:=\text{mean~}\|X\beta-X\betatrue\|_2^2/n,\\
\text{est. error}&:=\text{mean~}\|\beta-\betatrue\|_2,\\
\text{false neg.}&:=\text{mean~}|\{j\in S, j\notin S(\beta)\}|,\\
\text{false pos.}&:=\text{mean~}|\{j\notin S, j\in S(\beta)\}|.
\end{align*}
Upper bounds on the empirical standard deviations of these quantities are also given (in the brackets next to the corresponding results). Finally, we note that we used the tuning parameter $\lambda=2\sigma\sqrt{2\log(2p)}$ (cf. Remark~\ref{rm:tuningparameter}) and the glmnet algorithm \cite{Hastie10} (see Section~\ref{sec:numerics} for more details on the computation).
\begin{center}
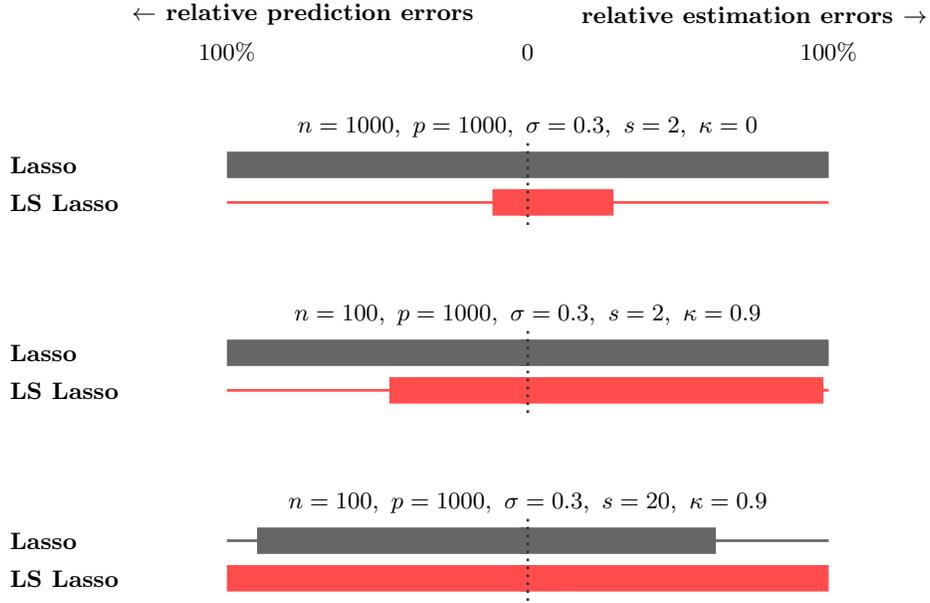
\begin{figure}[t]
  \begin{tikzpicture}{scale=0.1}

\node[align=center] at (0,5) {\normalfont{0}};
\node[align=center] at (4,5) {$100 \%$};
\node[align=center] at (-4,5) {$100 \%$};

\node[right] at (0.6,5.5) {\bf relative estimation errors $\rightarrow$};
\node[left] at (-0.6,5.5) {\bf $\leftarrow$  relative prediction errors};

\node[align=center] at (0,4) {$n=1000,~p=1000,~\sigma=0.3,~s=2,~\kappa=0$};
\node[align=center] at (0,1.5) {$n=100,~p=1000,~\sigma=0.3,~s=2,~\kappa=0.9$};
\node[align=center] at (0,-1) {$n=100,~p=1000,~\sigma=0.3,~s=20,~\kappa=0.9$};

\draw [red!70!white, line width=1] (-4,3) -- (4,3);

\draw [red!70!white, line width=10] (0,3) -- (-0.47,3);
\draw [black!60!white!, line width=10] (-0,3.5) -- (-4,3.5);
\draw [red!70!white, line width=10] (--1.14,3) -- (-0.1,3);
\draw [black!60!white, line width=10] (--4,3.5) -- (-0.1,3.5);

\draw [red!70!white, line width=1] (--4,0.5) -- (-4,0.5);

\draw [red!70!white, line width=10] (-0,0.5) -- (-1.84,0.5);
\draw [black!60!white, line width=10] (-0,1) -- (-4,1);
\draw [red!70!white, line width=10] (--3.93,0.5) -- (-0.1,0.5);
\draw [black!60!white, line width=10] (--4,1) -- (-0.1,1);

\draw [black!60!white, line width=1] (--4,-1.5) -- (-4,-1.5);
%\draw [red!70!white, line width=1] (--4,-2.5) -- (-4,-2.5);

\draw [red!70!white, line width=10] (-0,-2) -- (-4,-2);
\draw [black!60!white, line width=10] (-0,-1.5) -- (-3.6,-1.5);
\draw [red!70!white, line width=10] (--4,-2) -- (-0.1,-2);
\draw [black!60!white, line width=10] (--2.5,-1.5) -- (-0.1,-1.5);

\node[right] at (-7,3.5) {\bf Lasso};
\node[right] at (-7,3) {\bf LS Lasso};

\node[right] at (-7,1) {\bf Lasso};
\node[right] at (-7,0.5) {\bf LS Lasso};

\node[right] at (-7,-1.5) {\bf Lasso};
\node[right] at (-7,-2) {\bf LS Lasso};

\draw[black!80!white, line width = 1, dotted] (0,2.7) -- (0,3.8);
\draw[black!80!white, line width = 1, dotted] (0,0.2) -- (0,1.3);
\draw[black!80!white, line width = 1, dotted] (0,-2.3) -- (0,-1.2);

\end{tikzpicture}
\caption{Relative prediction and estimation errors of the Lasso~\eqref{eq:lasso} and the LS~Lasso (which is the corresponding LS refitted estimator~\eqref{eq:estimatorrefitted}) according to Table~\ref{tab:fo}.}\label{fig:foillu}
\end{figure}
\end{center}
\paragraph{First observations} The performances of the estimators for three sets of parameters are reported in Table~\ref{tab:fo}; the corresponding relative prediction and estimation errors are additionally plotted in Figure~\ref{fig:foillu}.\\
We first note that the variable selection performances of the Lasso and the LS~Lasso were equal for all sets of parameters. This is theoretically justified in Theorem~\ref{res:main} in the following section.\\
Let us now have a look at the prediction and estimation errors for the three sets of parameters. The set of parameters that corresponds to the top of Table~\ref{tab:fo} contains a small number of pertinent parameters $s$, a small level of correlation $\kappa$, and a large sample size $n$. Consistent with the literature, the Lasso performed well for both prediction and estimation~\cite{Buhlmann11}, and the LS~Lasso outmatched the Lasso regarding prediction~\cite{Belloni09}. Additionally, we observe that the LS Lasso also outmatched the Lasso regarding estimation.\\
The set of parameters that corresponds to the center of Table~\ref{tab:fo} contains a small number of pertinent parameters but highly correlated design matrices. We observe  that the Lasso performed only slightly better than the zero estimator for estimation but still exhibited good prediction performances, both consistent with the literature \cite{YoyoMomo12, vdGeer11}. Nevertheless, the Lasso was again outperformed by the LS Lasso, particularly for prediction.\\
Finally, the set of parameters that corresponds to the bottom of Table~\ref{tab:fo} contains a larger number of pertinent parameters and highly correlated design matrices. The Lasso performed well for prediction but only slightly better than the Zero estimator for estimation. More importantly, the LS Lasso performed worse than the Lasso. This was particularly the case for estimation, where the LS Lasso performed even worse than the Zero estimator.

\paragraph{Conclusions} The differences among the three sets of parameters  indicate that the benefit of least-squares refitting crucially depends on the setting. Moreover, the different results for prediction and estimation indicate that, for some settings, least-squares refitting may be beneficial for prediction but disadvantageous for estimation.\\
These observations motivate the remainder of this paper: In Section~\ref{sec:theory}, we present 
\begin{itemize}
\item bounds that hold for arbitrary designs and relate of the prediction and estimation errors of LS refitted estimators~\eqref{eq:estimatorrefitted} and the ones of the corresponding initial estimators~\eqref{eq:estimator}.
\end{itemize}
In Section~\ref{sec:criteria}, we introduce and study
\begin{itemize}
\item a criterion (based only on accessible quantities) to determine whether least-squares refitting is beneficial for a given setting and task.
\end{itemize}
All results are then summarized in Section~\ref{sec:discussion}.

\subsection{General error bounds for LS refitted estimators}\label{sec:theory}
In this section, we relate the errors of the LS~refitted estimators with the errors of the associated initial estimators. In particular, we introduce relations for both prediction and estimation that hold for arbitrary settings and are especially useful if the level of correlation in the design matrix is high. They indicate - in accordance with the simulations results in the previous section - that least-squares refitting is beneficial for some settings but disadvantageous for others.

\begin{theorem}\label{res:main} With probability one, the LS refitted estimator~\eqref{eq:estimatorrefitted} relates to the initial estimator~\eqref{eq:estimator} as follows: 
  \begin{equation*}
    \setref=\set,
  \end{equation*}
  \begin{align*}
 \|\estref-\est\|_q&=\|\grasetinv\sign(\estset)\|_q\frac{\lambda}{2g'\left(\|Y-X\est\|_2^2\right)}
  \end{align*}
for any $q\in(0,\infty]$, and
\begin{align*}
&\|X\estref-X\betatrue\|_2^2-\|X\est-X\betatrue\|_2^2\leq\|\grasetinv\sign(\estset)\|_1\frac{\lambda\sigma\|\matset^T\epsilon\|_\infty}{g'\left(\|Y-X\est\|_2^2\right)}.
\end{align*}
\end{theorem}
%Theorem 2.1 states bounds for the losses that can result from the application of least-squares refitting. These bounds provide an insight into the benefits and pitfalls of least-squares refitting:

\begin{remark}[Generality] In contrast to the results for prediction in~\cite{Belloni09}, Theorem~\ref{res:main} holds for any design matrix $X$. In particular, it is not presumed that the correlations in the design matrix are bounded with respect to restricted eigenvalues or similar quantities. Correlations are common and deserve special attention: first, common measures for correlations such as 
restricted eigenvalues depend on the inaccessible sparsity level $s$ and can thus often not be evaluated or even reasonably bounded; it then remains unclear whether results that invoke such measures apply or not. Second, $\ell_1$-penalized methods can perform well even in highly correlated settings, see  Remark~\ref{rm:correlations} below and \cite{YoyoMomo12, vdGeer11}; therefore, a study of refitted estimators in such settings is also of interest.
\end{remark}

\begin{remark}[Oracle inequalities]
For low levels of correlation in the design matrix, the above relations can be seen as oracle inequalities, since the quantity $ \|\grasetinv\sign(\estset)\|_q$ can be bounded accordingly, see Remarks~\ref{rm:correlations} and~\ref{rm:tuningparameter}. 
These bounds for $ \|\grasetinv\sign(\estset)\|_q$ do not hold,  however, for higher levels of correlation; this reflects the possibly unfavorable effects of least-squares refitting.
\end{remark}

\begin{remark}[Different levels of correlation]\label{rm:correlations} Theorem~\ref{res:main} holds for any degree of correlation in the design matrix $X$. However, the correlations appear on the right hand sides of the second and third relation through the matrix $\grasetinv$ and influence therefore the bounds for the estimation and prediction errors of the LS refitted estimator:\\
If the design matrix is only very weakly correlated, LS refitting can be highly beneficial \cite{Belloni09}, since $\ell_1$-penalized methods are typically consistent in terms of variable selection in this case \cite{BuneaEN, Yoyo13, Karim08, BiYuConsistLasso}. This is also reflected to some extent in Theorem~\ref{res:main}: For example, the mutual coherence assumptions
\begin{equation*}
   \label{eq:mutcoh}
    |X_i^TX_j|\leq n/(2s)~\text{~~~for all~}i\in S,~ 1\leq j \leq p,~ j\neq i,
  \end{equation*}
ensures correct variable selection, that is $\set=S$, for the Lasso under weak conditions on the minimal nonzero entries of the regression vector $\betatrue$ \cite{BuneaEN, Karim08}. From Lemma~\ref{res:invertibility} in Section~\ref{sec:aux}, we can then deduce \begin{equation*}
  \|\grasetinv\sign(\estset)\|_q\leq 2s^{\frac 1 q}/n,
\end{equation*}
so that the bounds in Theorem~\ref{res:main} match the well-known ``fast rate" bounds for the prediction and estimation errors for weakly correlated design matrices \cite{Bickel09, Belloni11}. Theorem~\ref{res:main} implies therefore that the LS refitted estimator performs at least as well as the initial estimator in this case. We note, however, that - under such strict conditions on the correlations in the design matrix - more favorable bounds for the prediction errors of LS~refitted estimators are known~\cite{Belloni09}.\\
The situation is more involved, for both the initial estimator and the LS refitted estimator, if the correlations are not small. The initial estimator does not necessarily provide consistent variable selection in this case, and the fast rates bounds for the estimation and prediction errors do also not apply. Nevertheless, the initial estimator  can still have favorable properties, especially regarding prediction. This is reflected, for example, in so called ``slow rate" bounds for the prediction errors that involve smaller tuning parameters $\lambda$ but are proportional to $\lambda\|\betatrue\|_1$ instead of $s\lambda^2$ \cite{YoyoMomo12} (the term ``slow rate" is somewhat misleading, since the corresponding rates are not necessarily slow). According to Theorem~\ref{res:main}, the LS refitted estimator fulfills such slow rate bounds if
\begin{equation*}
 \lambda  \|\grasetinv\sign(\estset)\|_1\approx~\|\betatrue\|_1.
\end{equation*}
If instead   
\begin{equation*}
 \lambda  \|\grasetinv\sign(\estset)\|_1\gg \|\betatrue\|_1,
\end{equation*}
the bound for the prediction error for the LS refitted estimator is larger than the corresponding bounds for the initial estimator. This is consistent with the observations made in the previous section: The LS refitted estimator may be beneficial even in highly correlated settings, but this is not always the case, especially if the estimated sparsity level $\setnor$ is not very small (and thus $\grasetinv$ is not a very small matrix).
\end{remark}

\begin{remark}[Tuning parameters]\label{rm:tuningparameter}
The tuning parameters $\lambda$ typically fulfill
  \begin{equation*}
    \label{eq:choicelambda}
 \lambda\gtrsim g'\left(\sigma\|\epsilon\|_2^2\right)\sigma\|X^T\epsilon\|_\infty 
  \end{equation*}
with high probability \cite{Bickel09,Belloni11,Buhlmann11}. Moreover, the estimator $X\est$ is typically consistent for $X\betatrue$ under weak assumptions (for the Lasso, see for example \cite{HCB08,Kolt10,MM11,RigTsy11}). Then,
\begin{equation*}
  \frac{\lambda}{g'\left(\|Y-X\est\|_2^2\right)}\approx \frac{\lambda}{g'\left(\sigma\|\epsilon\|_2^2\right)}
\end{equation*}
 and
\begin{equation*}
  \frac{\lambda\sigma\|X^T\epsilon\|_\infty}{g'\left(\|Y-X\est\|_2^2\right)}\lesssim \frac{\lambda^2}{\left(g'\left(\sigma\|\epsilon\|_2^2\right)\right)^2}.
\end{equation*}  
\end{remark}

\begin{remark}[Infeasible criteria]
  The two last relations in Theorem~\ref{res:main} indicate potential pitfalls of least-squares refitting:  One the one hand, if the last term of a relation is small, the LS~refitted estimators  perform - for the task corresponding to the relation -  at least comparable to the associated initial estimators and can therefore be used safely. On the other hand, if the last term of a relation is large, the LS~refitted estimators can perform considerably worse than the associated initial estimators, and therefore, their use is perhaps not advisable. One might thus be tempted to use the ratios of the two terms on the right hand sides of the relations as criteria for the application of least-squares refitting; unfortunately, while the quantity $ \|\grasetinv\sign(\estset)\|_q$ can be readily computed, the errors of the initial estimators cannot be computed or reasonably bounded since that would involve inaccessible parameters such as $\betatrue$ or $S$ (cf. Section~\ref{sec:fo}).
\end{remark}

\begin{remark}[Alternative refitting procedures]\label{rm:othermethods}
Besides the LS refitting estimators considered in this paper, the Adaptive Lasso and the Thresholded Lasso are well known multistage procedures, see \cite{Candes08, Meinsthresh09, SaraShu11} and references therein. In contrast to the LS refitting estimator, these procedures are designed to diminish the false positives of the Lasso.  In \cite{SaraShu11}, the Adaptive Lasso and the Thresholded Lasso are compared with the Lasso by means of involved oracle inequalities for settings with sufficiently small levels of correlation. The authors conclude that the Adaptive Lasso and the Thresholded Lasso perform in these settings comparably to the Lasso in terms of prediction and estimation but can outperform the Lasso in terms of variable selection.
\end{remark}

\subsection{A criterion}\label{sec:criteria} In this section, we introduce a criterion to distinguish between settings that are suited for least-squares refitting and settings that are problematic for least-squares refitting. From this criterion originates an estimator that adopts the outcome of - depending on the value of the criterion - either the initial estimator or the LS refitted estimator. For the case of the Lasso~\eqref{eq:lasso} and the corresponding LS refitted estimator~\eqref{eq:estimatorrefitted}, we perform a numerical study, which indicates that the estimator based on the criterion performs close to the better one of the two original estimators.

\paragraph{The criterion} Our goal is to detect whether least-squares refitting is beneficial in a setting under consideration. For this, we introduce the criterion~$F$ as a random variable that compares the signs of the vectors $\estset$ and $\grasetinv\sign(\estset)$:
\begin{equation}\label{eq:criterion}
  F(\set):=\frac{1}{\setnor}\left|\left\{j\in\set:\sign(\est_j)\neq \sign\left((\grasetinv\sign(\estset))_j\right)\right\}\right|\in[0,1].
\end{equation}
The involvement of the vector $\grasetinv\sign(\estset)$ is no surprise regarding the theoretical results in Section~\ref{sec:theory}; the criterion is further motivated later, in Remark~\ref{rm:heuristics} in Section~\ref{sec:proofs}, so that we can focus  on its practical aspects in the following. Least-squares refitting is claimed to be beneficial if the criterion is small and disadvantageous if the criterion is large. The criterion depends only on $\est$ and $X$ and especially not on possibly inaccessible parameters such as $s$ and $\sigma$. Its computation is undemanding since only a matrix inversion of a regular and typically small matrix is needed.\\
The criterion prompts the application of the $c$-LS refitted estimator 
\begin{equation}
  \label{eq:estimatorcriterion}
  \estcri:=\begin{cases}\estref&\text{~~~~if $F(\set)\leq c$}\\\est&\text{~~~~otherwise}\end{cases}
\end{equation}
for a fixed $c\in [0,1]$. We note that the $c$-LS refitted estimator invokes the LS refitted estimator~\eqref{eq:estimatorrefitted} if the criterion suggests that the setting is suitable for it and invokes the initial estimator~\eqref{eq:estimator} otherwise. The parameter~$c$ is the corresponding threshold: the larger the value of $c$, the more likely refitting is applied. The findings in Section~\ref{sec:fo} (and also of the remainder of this section) indicate that the threshold should be lower for prediction than for estimation. Apart from this, we consider fixed, constant values for $c$ (so that especially no parameter tuning is involved).

\paragraph{Setting and computation} We studied the Lasso \eqref{eq:lasso} and the corresponding LS refitted estimator \eqref{eq:estimatorrefitted}, which we call LS Lasso, and the $c$-LS refitted estimator \eqref{eq:estimatorcriterion}, which we call $c$-LS Lasso. As thresholds, we used for all simulations the constants $c=0.4$ for prediction and $c=0.2$ for estimation. We invoked settings as described in Section~\ref{sec:fo}. The computations were also done as described in Section~\ref{sec:fo}, except for the additional tracking of the application of least-squares refitting through 
\begin{align*}
  \text{LS pred.}&:=\frac{\text{number of repetitions involving refitting for prediction}}{\text{repetitions}}\\
\text{LS est.}&:=\frac{\text{number of repetitions involving refitting for estimation}}{\text{repetitions}}
\end{align*}
for each estimator. The Lasso and the Zero estimator \eqref{eq:zero} do never involve least-squares refitting (LS pred.=LS est.=$0$), the LS Lasso always involves least-squares refitting (LS pred.=LS est.=$1$), and the number of applications where least-squares refitting is applied in the $c$-LS Lasso depends on the values of the criterion~\eqref{eq:criterion} and the threshold $c$.

\paragraph{Results} We considered the three sets of parameters  used in Section~\ref{sec:fo} and three additional ones. The corresponding performances of the estimators are reported in Table~\ref{tab:data}; the relative prediction and estimation errors are additionally plotted in Figure~\ref{fig:dataillu}.\\
Comparing the performances of the Lasso, the LS Lasso, and the $c$-LS Lasso with the performances of the Zero estimator, differences between the tasks become visible: On the one hand, all three estimators exhibited good prediction performances for all sets of parameters (see the first column in Table~\ref{tab:data}). On the other hand, their success with respect to estimation and variable selection depended on the parameters, especially on the levels of correlation $\kappa$ (columns two, three, and four). We note again that the levels of correlation are known in these simulations but are usually not accessible in practice.\\ 
For small sparsity levels $s$ (see the sets of parameters  one and two as counted from the top in Table~\ref{tab:data} and Figure~\ref{fig:dataillu}) or small to medium levels of correlation (sets of parameters one and five), the least-squares refitting was beneficial for prediction and beneficial or only slightly disadvantageous for estimation. In contrast, for both larger sparsity levels and larger levels of correlation (sets of parameters three and six), the least-squares refitting was disadvantageous for both prediction and estimation. One also observes that the least-squares refitting was beneficial for a wider range of parameters regarding prediction as opposed to estimation (sets of parameters four, for example).\\
For all sets of parameters, the fractions of the application of least-squares refitting for the $c$-LS Lasso (last column) reflected the benefit of least-squares refitting. The $c$-LS Lasso was therefore close to either the Lasso or the LS~Lasso - whichever was better for the setting and the task under consideration.

\paragraph{Conclusions} The numerical results support the conclusion of Section~\ref{sec:fo} that the benefit of least-squares refitting  crucially depends on the setting and the task under consideration.\\
Additionally, the results suggest that the criterion~\eqref{eq:criterion} can serve as a tool to determine the usefulness - or disutility - of least-squares refitting.

\subsection{Summary}\label{sec:discussion}
Both the theoretical and the numerical findings in Section~\ref{sec:main} indicate that whether least-squares refitting should be applied depends on the setting and the task under consideration. First, the relations in Theorem~\ref{res:main} (see Section~\ref{sec:theory}) and the simulation results in Table~\ref{tab:data} and Figure~\ref{fig:dataillu}  (see Section~\ref{sec:criteria}) confirm the usefulness of least-squares refitting for prediction with mildly correlated design matrices, cf.~\cite{Belloni09}.  The relations and simulation results additionally demonstrate that least-squares refitting can be advantageous for prediction with highly correlated design matrices and for estimation. However, they also reveal that least-squares refitting is problematic if the design matrix is correlated and the sparsity level is considerably larger than $1$. Moreover, the different simulation results for prediction and estimation, see Table~\ref{tab:data} and Figure~\ref{fig:dataillu}, indicate that least-squares refitting is more beneficial for prediction than for estimation; this is not surprising in view of its definition (see~\eqref{eq:estimatorrefitted} in the introduction), which reflects that least-squares refitting was designed for reducing prediction biases. Finally, the good performances of the $c$-LS Lasso in the simulation study in Section~\ref{sec:criteria} suggest the use of the random variable $F$ (see~\eqref{eq:criterion} in Section~\ref{sec:criteria}) as a criterion for the application of least-squares refitting.\\

We note that the simulations were restricted to the Lasso~\eqref{eq:lasso} and to a standard normal distribution for the noise vector $\epsilon$ in the model~\eqref{eq:model}, since they are by far the most common estimator and distribution of the noise vector, respectively, in our framework. Simulations for other estimators and distributions of the noise vector are of interest but beyond the scope of this paper.

\begin{table}
\begin{tabular}{l c c c c c}
 & \bf{pred. error} & \bf est. error & \bf false neg. & \bf false pos. & \bf LS pred./est.\\ 
\multicolumn{6}{c}{~}\\
\multicolumn{6}{c}{$n=1000,~p=1000,~\sigma=0.3,~s=2,~\kappa=0$}\vspace{0.2cm}\\
{\bf Lasso} & $(2.87 \pm 0.04)\times 10^{\text -3}$ & $(5.28\pm 0.03)\times 10^{\text -2}$ & $0.00\pm 0.00$ & $0.11\pm 0.01$ &  0/0 \\ 
{\bf LS Lasso}& $(0.34 \pm 0.02)\times 10^{\text -3}$ & $(1.50\pm 0.04)\times 10^{\text -2}$ & $0.00\pm 0.00$ & $0.11\pm 0.01$ &  1/1 \\
{\bf $c$-LS Lasso}& $(0.34 \pm 0.02)\times 10^{\text -3}$ & $(1.50\pm 0.04)\times 10^{\text -2}$ & $0.00\pm 0.00$ & $0.11\pm 0.01$ &  1/1\\
{\bf Zero}&$1.25\pm 0.01$&$1.12$&2&0&0/0\\
~&~&~&~&~&~\\
\multicolumn{6}{c}{~}\\
\multicolumn{6}{c}{$n=100,~p=1000,~\sigma=0.3,~s=2,~\kappa=0.9$}\vspace{0.2cm}\\
{\bf Lasso} & $(2.31 \pm 0.03)\times 10^{\text -2}$ & $(9.38\pm 0.07)\times 10^{\text -1}$ & $0.87\pm 0.02$ & $11.6\pm 0.2$ & 0/0 \\ 
{\bf LS Lasso}& $(1.06 \pm 0.02)\times 10^{\text -2}$ & $(9.22\pm 0.07)\times 10^{\text -1}$ & $0.87\pm 0.02$ & $11.6\pm 0.2$ & 1/1 \\
{\bf $c$-LS Lasso}& $(1.06 \pm 0.02)\times 10^{\text -2}$ & $(9.26\pm 0.07)\times 10^{\text -1}$ & $0.87\pm 0.02$ & $11.6\pm 0.2$ & 0.997/0.741\\
{\bf Zero}&$2.24\pm 0.01$&$1.12$&2&0&0/0\\
~&~&~&~&~&~\\ 
\multicolumn{6}{c}{~}\\
\multicolumn{6}{c}{$n=100,~p=1000,~\sigma=0.3,~s=20,~\kappa=0.9$}\vspace{0.2cm}\\
{\bf Lasso} & $(5.10 \pm 0.04)\times 10^{\text -2}$ & $2.19\pm 0.02$ & $8.29\pm 0.08$ & $50.9\pm 0.4$ & 0/0 \\ 
{\bf LS Lasso}& $(5.67 \pm 0.04)\times 10^{\text -2}$ & $3.51\pm 0.04$ & $8.29\pm 0.08$ & $50.9\pm 0.4$ & 1/1 \\
{\bf $c$-LS Lasso}& $(5.15 \pm 0.04)\times 10^{\text -2}$ & $2.19\pm 0.02$ & $8.29\pm 0.08$ & $50.9\pm 0.4$ & 0.12/0\\
{\bf Zero}&$108.99\pm 0.01$&$2.68$&20&0&0/0\\
~&~&~&~&~&~\\
\multicolumn{6}{c}{~}\\
\multicolumn{6}{c}{$n=100,~p=1000,~\sigma=0.1,~s=10,~\kappa=0.9$}\vspace{0.2cm}\\
{\bf Lasso} & $(5.65 \pm 0.07)\times 10^{\text -3}$ & $(7.31\pm 0.08)\times 10^{\text -1}$ & $1.75\pm 0.03$ & $40.5\pm 0.5$ &  0/0 \\ 
{\bf LS Lasso}& $(5.31 \pm 0.05)\times 10^{\text -3}$ & $(9.66\pm 0.12)\times 10^{\text -1}$ & $1.75\pm 0.03$ & $40.5\pm 0.5$ &  1/1 \\
{\bf $c$-LS Lasso}& $(5.60 \pm 0.06)\times 10^{\text -3}$ & $(7.32\pm 0.08)\times 10^{\text -1}$ & $1.75\pm 0.03$ & $40.5\pm 0.5$ &  0.36/0.013\\
{\bf Zero}&$29.92\pm 0.01$&$1.96$&10&0&0/0\\
~&~&~&~&~&~\\
\multicolumn{6}{c}{~}\\
\multicolumn{6}{c}{$n=1000,~p=1000,~\sigma=0.1,~s=20,~\kappa=0.5$}\vspace{0.2cm}\\
{\bf Lasso} & $(9.53 \pm 0.10)\times 10^{\text -4}$ & $(4.01\pm 0.03)\times 10^{\text -2}$ & $0.00\pm 0.00$ & $49.8\pm 0.5$ & 0/0 \\ 
{\bf LS Lasso}& $(8.46 \pm 0.06)\times 10^{\text -4}$ & $(4.20\pm 0.02)\times 10^{\text -2}$ & $0.00\pm 0.00$ & $49.8\pm 0.5$ & 1/1 \\
{\bf $c$-LS Lasso}& $(8.46 \pm 0.06)\times 10^{\text -4}$ & $(4.16\pm 0.03)\times 10^{\text -2}$ & $0.00\pm 0.00$ & $49.8\pm 0.5$ & 1/0.266\\
{\bf Zero}&$58.62\pm 0.04$&$2.68$&20&0&0/0\\
~&~&~&~&~&~\\
\multicolumn{6}{c}{~}\\
\multicolumn{6}{c}{$n=100,~p=1000,~\sigma=1,~s=60,~\kappa=0.9$}\vspace{0.2cm}\\
{\bf Lasso} & $(5.39\pm 0.04)\times 10^{\text -1}$ & $\textcolor{white}{0}6.21\pm 0.02$ & $38.9\pm 0.2$ & $52.1\pm 0.4$ & 0/0 \\ 
{\bf LS Lasso}& $(6.94 \pm 0.04)\times 10^{\text -1}$ & $13.11\pm 0.25$ & $38.9\pm 0.2$ & $52.1\pm 0.4$ & 1/1 \\
{\bf $c$-LS Lasso}& $(5.44 \pm 0.04)\times 10^{\text -1}$ & $\textcolor{white}{0}6.21\pm 0.02$ & $38.9\pm 0.2$ & $52.1\pm 0.4$ & 0.039/0\\
{\bf Zero}&$919.02\pm 0.06$&$4.53$&60&0&0/0\\
~&~&~&~&~&~\\
\end{tabular}
\caption{Prediction, estimation, and variable selection performances of the Lasso~\eqref{eq:lasso}, the LS~Lasso (which is the corresponding LS refitted estimator~\eqref{eq:estimatorrefitted}), the $c$-LS Lasso (which is the corresponding $c$-LS refitted estimator~\eqref{eq:estimatorcriterion}), and the Zero estimator $0\in\rp$. The setting and computation are detailed in the section. For a plot of the relative prediction and estimation errors, see Figure~\ref{fig:dataillu} below.}\label{tab:data}
\end{table}

\begin{center}
\begin{figure}[h!]
  \begin{tikzpicture}{scale=0.1}

\node[align=center] at (0,5) {\normalfont{0}};
\node[align=center] at (4,5) {$100 \%$};
\node[align=center] at (-4,5) {$100 \%$};

\node[right] at (0.6,5.5) {\bf relative estimation errors $\rightarrow$};
\node[left] at (-0.6,5.5) {\bf $\leftarrow$  relative prediction errors};

\node[align=center] at (0,4) {$n=1000,~p=1000,~\sigma=0.3,~s=2,~\kappa=0$};
\node[align=center] at (0,1.5) {$n=100,~p=1000,~\sigma=0.3,~s=2,~\kappa=0.9$};
\node[align=center] at (0,-1) {$n=100,~p=1000,~\sigma=0.3,~s=20,~\kappa=0.9$};
\node[align=center] at (0,-3.5) {$n=100,~p=1000,~\sigma=0.1,~s=10,~\kappa=0.9$};
\node[align=center] at (0,-6) {$n=1000,~p=1000,~\sigma=0.1,~s=20,~\kappa=0.5$};
\node[align=center] at (0,-8.5) {$n=100,~p=1000,~\sigma=1,~s=60,~\kappa=0.9$};

\draw [red!70!white, line width=1] (-4,3) -- (4,3);
\draw [blue!60!white, line width=1] (-4,2.5) -- (4,2.5);

\draw [red!70!white, line width=10] (0,3) -- (-0.47,3);
\draw [black!60!white!, line width=10] (0,3.5) -- (-4,3.5);
\draw [red!70!white, line width=10] (1.14,3) -- (-0.1,3);
\draw [black!60!white, line width=10] (4,3.5) -- (-0.1,3.5);
\draw [blue!70!white, line width=10] (1.14,2.5) -- (-0.1,2.5);
\draw [blue!70!white, line width=10] (0,2.5) -- (-0.47,2.5);

\draw [red!70!white, line width=1] (-4,0.5) -- (4,0.5);
\draw [blue!60!white, line width=1] (-4,0) -- (4,0);

\draw [red!70!white, line width=10] (0,0.5) -- (-1.84,0.5);
\draw [black!60!white, line width=10] (0,1) -- (-4,1);
\draw [red!70!white, line width=10] (--3.93,0.5) -- (-0.1,0.5);
\draw [black!60!white, line width=10] (--4,1) -- (-0.1,1);
\draw [blue!70!white, line width=10] (--3.95,0) -- (-0.1,0);
\draw [blue!70!white, line width=10] (0,0) -- (-1.84,0);

\draw [black!60!white, line width=1] (-4,-1.5) -- (4,-1.5);
%\draw [red!70!white, line width=1] (-4,-2.5) -- (4,-2.5);
\draw [blue!60!white, line width=1] (-4,-2.5) -- (4,-2.5);

\draw [red!70!white, line width=10] (-0,-2) -- (-4,-2);
\draw [black!60!white, line width=10] (-0,-1.5) -- (-3.6,-1.5);
\draw [red!70!white, line width=10] (--4,-2) -- (-0.1,-2);
\draw [black!60!white, line width=10] (--2.5,-1.5) -- (-0.1,-1.5);
\draw [blue!70!white, line width=10] (--2.5,-2.5) -- (-0.1,-2.5);
\draw [blue!70!white, line width=10] (-0,-2.5) -- (-3.6,-2.5);

\draw [black!60!white, line width=1] (4,-4) -- (-0,-4);
\draw [red!70!white, line width=1] (-0,-4.5) -- (-4,-4.5);
\draw [blue!60!white, line width=1] (4,-5) -- (-4,-5);

\draw [red!70!white, line width=10] (-0,-4.5) -- (-3.76,-4.5);
\draw [black!60!white, line width=10] (-0,-4) -- (-4,-4);
\draw [red!70!white, line width=10] (--4,-4.5) -- (-0.1,-4.5);
\draw [black!60!white, line width=10] (--3.03,-4) -- (-0.1,-4);
\draw [blue!70!white, line width=10] (--3.03,-5) -- (-0.1,-5);
\draw [blue!70!white, line width=10] (-0,-5) -- (-3.96,-5);

\draw [black!60!white, line width=1] (--4,-6.5) -- (-0,-6.5);
\draw [red!70!white, line width=1] (--0,-7) -- (-4,-7);
\draw [blue!60!white, line width=1] (--4,-7.5) -- (-4,-7.5);

\draw [red!70!white, line width=10] (-0,-7) -- (-3.55,-7);
\draw [black!60!white, line width=10] (-0,-6.5) -- (-4,-6.5);
\draw [red!70!white, line width=10] (--4,-7) -- (-0.1,-7);
\draw [black!60!white, line width=10] (--3.82,-6.5) -- (-0.1,-6.5);
\draw [blue!70!white, line width=10] (--3.96,-7.5) -- (-0.1,-7.5);
\draw [blue!70!white, line width=10] (-0,-7.5) -- (-3.55,-7.5);

%\draw [red!80!white, line width=10] (0,-7) -- (8.44,-7);

\draw [black!60!white, line width=1] (--4,-9) -- (-4,-9);
%\draw [red!70!white, line width=1] (--4,-9.5) -- (-4,-9.5);
\draw [blue!60!white, line width=1] (--4,-10) -- (-4,-10);

\draw [red!70!white, line width=10] (-0,-9.5) -- (-4,-9.5);
\draw [black!60!white, line width=10] (-0,-9) -- (-3.11,-9);
\draw [red!70!white, line width=10] (--4,-9.5) -- (-0.1,-9.5);
\draw [black!60!white, line width=10] (--1.89,-9) -- (-0.1,-9);
\draw [blue!70!white, line width=10] (--1.89,-10) -- (-0.1,-10);
\draw [blue!70!white, line width=10] (-0,-10) -- (-3.14,-10);

\node[right] at (-7,3.5) {\bf Lasso};
\node[right] at (-7,3) {\bf LS Lasso};
\node[right] at (-7,2.5) {\bf $c$-LS Lasso};

\node[right] at (-7,1) {\bf Lasso};
\node[right] at (-7,0.5) {\bf LS Lasso};
\node[right] at (-7,0) {\bf $c$-LS Lasso};

\node[right] at (-7,-1.5) {\bf Lasso};
\node[right] at (-7,-2) {\bf LS Lasso};
\node[right] at (-7,-2.5) {\bf $c$-LS Lasso};

\node[right] at (-7,-4) {\bf Lasso};
\node[right] at (-7,-4.5) {\bf LS Lasso};
\node[right] at (-7,-5) {\bf $c$-LS Lasso};

\node[right] at (-7,-6.5) {\bf Lasso};
\node[right] at (-7,-7) {\bf LS Lasso};
\node[right] at (-7,-7.5) {\bf $c$-LS Lasso};

\node[right] at (-7,-9) {\bf Lasso};
\node[right] at (-7,-9.5) {\bf LS Lasso};
\node[right] at (-7,-10) {\bf $c$-LS Lasso};

\draw[black!80!white, line width = 1, dotted] (0,2.2) -- (0,3.8);
\draw[black!80!white, line width = 1, dotted] (0,-0.3) -- (0,1.3);
\draw[black!80!white, line width = 1, dotted] (0,-2.8) -- (0,-1.2);
\draw[black!80!white, line width = 1, dotted] (0,-5.3) -- (0,-3.7);
\draw[black!80!white, line width = 1, dotted] (0,-7.8) -- (0,-6.2);
\draw[black!80!white, line width = 1, dotted] (0,-10.3) -- (0,-8.7);

\end{tikzpicture}
\caption{Relative prediction and estimation errors of the Lasso~\eqref{eq:lasso}, the LS~Lasso (which is the corresponding LS refitted estimator~\eqref{eq:estimatorrefitted}), and the $c$-LS Lasso (which is the corresponding $c$-LS refitted estimator~\eqref{eq:estimatorcriterion}) according to Table~\ref{tab:data}.}\label{fig:dataillu}
\end{figure}
\end{center}

\clearpage

\section{Proofs and further remarks}\label{sec:proofs}

\subsection{Auxiliary results}\label{sec:aux}

\begin{lemma}\label{res:lassoinvertible}
  The design matrix $X$ restricted to the active set $\set$ of the estimator~\eqref{eq:estimator} has full rank. In other words: The matrix $\graset\in\R^{\setnor\times \setnor}$ is invertible.
\end{lemma}

\begin{lemma}\label{res:lsfitting}
Let $A\subset\{1,\dots,p\}$ be a (possibly random) nonempty set such that $X_A^TX_A\in \R^{|A|\times|A|}$ is invertible and let 
\begin{equation*}
    \estref_A:=\argmin_{\beta\in\R^{|A|}}\|Y-X_A\beta\|_2^2
\end{equation*}
be the least-squares estimator on the set $A$. Then, $S(\estref_A)=A$ with probability one.
\end{lemma}

\begin{lemma}\label{res:kkt}
%Assume that  $Y\neq X\widehat\beta$.
The vector $\widehat\beta$ is a solution of the criterion \eqref{eq:estimator} if and only if for every  $1\leq j\leq p$
\begin{align*}
&\widehat\beta_j\neq 0 \Rightarrow (X^T(Y-X\widehat\beta))_j=\frac{\lambda}{2g'\left(\|Y-X\est\|_2^2\right)}\sign(\est_j),\\% \label{kkt1}\\
&\widehat\beta_j= 0\Rightarrow |(X^T(Y-X\est))_j|\leq \frac{\lambda }{ 2g'\left(\|Y-X\est\|_2^2\right)}. % \label{kkt2}
\end{align*}
\end{lemma}

\begin{lemma}\label{res:invertibility}
  Assume that for a nonempty set $A\subset \indexlasso$ it holds that $|X_i^TX_j|\leq \frac n {2|A|}$ for all $i,j\in A$ with $i\neq j$. Then, $X_A^TX_A$ is invertible and its inverse $(X_A^TX_A)^{\text -1}$ fulfills the following inequalities for all $l\in A$:
  \begin{align*}
\sum_{\substack{k=1\\k\neq l}}^{|A|}|\left((X_A^TX_A)^{\text -1}\right)_{kl}|\leq \frac 1 n-\frac 1 {n|A|}\leq ((X_A^TX_A)^{\text -1})_{ll}\leq \frac 1 n +\frac 1 {n|A|}.
  \end{align*}
These inequalities become strict inequalities if $|X_i^TX_j|< \frac n {2|A|}$ for all $i,j\in A$ with $i\neq j$.
\end{lemma}

\subsection{Proofs of the auxiliary results}

\begin{proof}[Proof of Lemma~\ref{res:lassoinvertible}]
We proof the claim by contradiction. To this end, we assume that the matrix $\graset$ is not invertible. This means that there exists a nonzero vector $r\in\rset$ such that $\graset r=0$ and thus $\matset r =0$. Without loss of generality, we may assume that $\set= \{1,\dots, \setnor\}$ and therefore $X(r^T,0,\dots,0)^T=\matset r=0$. Consequently, for any scalar $\alpha\in\R$ and associated vector $v_\alpha:=\est + \alpha (r^T,0,\dots,0)^T\in\rp$, it holds that 
\begin{equation}
  \label{eq:invertb1}
g(Y-Xv_\alpha)=g(Y-X\est).  
\end{equation}
Next, we note that $\estset$ belongs to the boundary of the compact set $K:=\{u\in\rset:\|u\|_1\leq\|\estset\|_1\}\subset \R^\setnor$. Simple geometric considerations reveal that there is a scalar $\alpha\in\R$ such that the vector $w_\alpha:=\estset+\alpha r\in\rset$ has the properties $(w_\alpha)_i=0$ for an index $i\in\{1,\dots,\setnor\}$ and $w_\alpha\in K$. For this specific value of $\alpha$, the first property implies
\begin{equation*}
  \label{eq:invert2}
  \|v_\alpha\|_0=\|w_\alpha\|_0<\setnor=\|\est\|_0, 
\end{equation*}
where $\|\cdot\|_0$ denotes the number of nonzero entries of a vector. Additionally, the second property implies $\|v_\alpha\|_1=\|w_\alpha\|_1\leq\|\estset\|_1=\|\est\|_1$ and therefore, together with Equation~\eqref{eq:invertb1}, 
\begin{equation*}
  \label{eq:invert3}
  g(Y-Xv_\alpha)+\lambda\|v_\alpha\|_1\leq g(Y-X\est)+\lambda\|\est\|_1.
\end{equation*}
The last two displays contradict the assumption in Section~\ref{sec:conventions} that $\est$ is a minimizer of~\eqref{eq:estimator} with a minimal number of nonzero entries among all minimizers.
\end{proof}

\begin{proof}[Proof of Lemma~\ref{res:lsfitting}] The proof relies on simple geometric considerations and the assumption on the noise stated in Section~\ref{sec:conventions}.\\
As a first step, we consider the problem for a fixed but arbitrary nonempty set $B\subset \indexlasso$. To this end, we define the subspace $V_B:=\op{span}\{X_j:j\in B\}$ as the linear span of the columns of $X$ with indices in $B$, similarly  $V_{B-\{j\}}:=\op{span}\{X_j:j\in B-\{j\}\}$, and we denote by % Moreover, we denote by $U_B:=\{u+X\betatrue:u\in V_B\}$ the subspace $V_B$ shifted by the vector $X\betatrue$.\\
$P_B$ the operator for the orthogonal projection on $V_B$. It holds that $Y=P_BY+(1-P_B)Y$, $P_BX_Bv=X_Bv$, and $P_BY, X_Bv\perp (1-P_B)Y$ for any vector $v\in\R^{|B|}$. Therefore,  
\begin{equation*}\label{eq:reformin1}
  \argmin_{\beta\in\R^{|B|}}\|Y-X_B\beta\|_2^2=\argmin_{\beta\in\R^{|B|}}\|P_BY-X_B\beta\|_2^2
\end{equation*}
and
\begin{equation*}\label{eq:reformin2}
  \min_{\beta\in\R^{|B|}}\|P_BY-X_B\beta\|_2^2=0.
\end{equation*}
From the last two displays, we deduce
\begin{equation*}
  \mpr(S(\estref_B)\neq B)=\mpr(\exists j\in B:P_BY\in V_{B-\{j\}}).
\end{equation*}
The claim is now a consequence of the union bound. To see this, note that $P_BY\in V_{B-\{j\}}\Leftrightarrow Y \in (V_B)^\perp\oplus V_{B-\{j\}}$. Moreover, if $X_B^TX_B$ is invertible, it holds that $(V_B)^\perp\oplus V_{B-\{j\}}\neq \rn$. We can finally invoke the union bound and the assumption on the noise to obtain
\begin{align*}
\mpr(S(\estref_A)\neq A)&\leq \sum_{B}  \mpr(S(\estref_B)\neq B)\\
&\leq \sum_{B}\sum_{j\in B}\mpr(Y \in (V_B)^\perp\oplus V_{B-\{j\}})\\
&=0,
\end{align*}
where the sum is taken over all nonempty sets $B\subset \indexlasso$ such that $X_B^TX_B$ is invertible.
\end{proof}

\begin{proof}[Proof of Lemma \ref{res:kkt}] 
KKT conditions \cite{kkt} have been derive for many situations, but for convenience, we  still provide a detailed derivation for the results needed in this paper. To this end, we generalize results in \cite{Yoyo13}, where the case of the Square-Root Lasso is treated.\\
Since all terms of the criterion \eqref{eq:estimator} are convex, and
thus, the criterion is convex, we can apply standard subgradient calculus. The subgradient $\partial_xf$ of a convex function $f:\R^p\to\R$ at a
point $x\in\R^p$ is defined as the
set of vectors $v\in\R^p$ such that for all $y\in\R^p$
\begin{equation*}
  f(y)\geq f(x)+v^T(y-x).
\end{equation*}
From this, one derives easily that subgradients are linear and additive and
that the subgradient $\partial_xf$ is equal to the gradient $\nabla_xf$ if
the function $f$ is differentiable at $x$. Moreover,  $x\in\R^p$ is a
minimum of the function $f$ if  and only if $0\in\partial_xf$.
The first term of the criterion \eqref{eq:estimator} is differentiable, and we have
\begin{align}\label{eq.yoyo.sub1}
\nonumber\nabla_\beta ~g\left(\|Y-X\beta\|_2^2\right)=&g'\left(\|Y-X\beta\|_2^2\right)\nabla_\beta\|Y-X\beta\|_2^2\\
=&\text -2g'\left(\|Y-X\beta\|_2^2\right)X^T(Y-X\beta).
\end{align}
For the remaining term, we observe that for any scalar $u\in\R\backslash\{0\}$
\begin{align}\label{eq.yoyo.sub2}
\nabla_u|u|=\sign(u).
\end{align}
Moreover, 
\begin{align}\label{eq.yoyo.sub3}
v\in\partial_{u=0}|u|\Leftrightarrow |z|\geq |0|+v^T(z-0)=v^Tz~~~\text{for
all~}z\in\R
\end{align}
and consequently $\partial_{u=0}|u|=\{v\in\R:|v|\leq 1\}$.\\
The claim follows then from Equations \eqref{eq.yoyo.sub1},
\eqref{eq.yoyo.sub2}, and  \eqref{eq.yoyo.sub3} by the additivity of subgradients.
\end{proof}

\newcommand{\sumkhere}{\sum_{k=1}^a}
\newcommand{\sumjhere}{\sum_{j=1}^a}
\newcommand{\sumihere}{\sum_{i=1}^a}
\newcommand{\sumlhere}{\sum_{l=1}^a}

\begin{proof}[Proof of Lemma~\ref{res:invertibility}]
The proof consists of the application of simple algebra. We only derive the inequalities corresponding to the first part, since the strict inequalities corresponding to the second part can be derived along the same lines.\\ 
First, for ease of exposition, we denote the cardinality of $A$ by $a:=|A|$ and introduce the matrix $B:=\1-X_A^TX_A/n\in \R^{a\times a}$. The entries of the matrix~$B$ are bounded according to 
\begin{equation}\label{eq:invert1}
|B_{ij}|=|\delta_{ij}-X_i^TX_j/n|\leq \frac {1-\delta_{ij}} {2a}\leq \frac{1}{2a}. 
\end{equation}
This leads for any $x\in \R^a$ to
  \begin{align*}
    |x^TBx|=|\sumihere x_i\sumjhere B_{ij}x_j|
\leq \sumihere |x_i|\sumjhere |B_{ij}||x_j|
%&\leq \frac{1}{2a} \sumihere |x_i|\sumjhere |x_j|\\
\leq\frac{\|x\|_1^2}{2a}
\leq \frac{\|x\|_2^2}{2}.
 \end{align*}
Hence, the largest singular value of the matrix $B$ is smaller or equal to $1/2$. The limit $\sum_{i=1}^\infty B^i$ therefore exists, and the inverse of $X_A^TX_A/n$ is given by
\begin{equation}
  \label{eq:invertible}
  (X_A^TX_A/n)^{\text -1}=(\1-B)^{\text -1}=\1+\sum_{i=1}^\infty B^i.
\end{equation}
We now have a look at the entries of $B^i$. For any $i\in\{1,2,\dots\}$ and $k,j\in A$, the bound~\eqref{eq:invert1} on the entries of $B$ implies
\begin{equation*}
  |(B^{i+1})_{kj}|=  |(B^{i}B)_{kj}|=|\sumlhere (B^{i})_{kl}B_{lj}|\leq \sumlhere |(B^{i})_{kl}||B_{lj}|\leq \frac 1 {2a}\sumlhere |(B^i)_{kl}|.
\end{equation*}
One can thus deduce by induction that $|(B^{i})_{kj}|\leq \frac 1 {2^ia}$ for all $i\in\{1,2,\dots\}$. Consequently,
\begin{align}\label{eq:invertibleentry}
  \sum_{i=1}^\infty|(B^i)_{kl}|\leq \sum_{i=1}^\infty\frac{1}{2^ia}= \frac 1 a
\end{align}
for all $k,l\in A$.\\
Now, we turn to the entries of $(X_A^TX_A/n)^{\text -1}$: Equations~\eqref{eq:invertible} and \eqref{eq:invertibleentry} yield
\begin{align*}
  \left((X_A^TX_A/n)^{\text -1}\right)_{ll}=1+\sum_{i=1}^\infty(B^i)_{ll}
\leq 1+\sum_{i=1}^\infty|(B^i)_{ll}| \leq 1+\frac 1 a
\end{align*}
for any $l\in A$. Similarly,
\begin{align*}
  \left((X_A^TX_A/n)^{\text -1}\right)_{ll}%=|1+\sum_{i=1}^\infty(B^i)_{ll}|
%\geq 1-\sum_{i=1}^\infty|(B^i)_{ll}|
\geq 1-\frac 1 a
\end{align*}
and
\begin{align*}
  \sum_{\substack{k=1\\k\neq l}}^a|\left((X_A^TX_A/n)^{\text -1}\right)_{kl}|&=\sum_{\substack{k=1\\k\neq l}}^a|\sum_{i=1}^\infty(B^i)_{kl}|
\leq \sum_{\substack{k=1\\k\neq l}}^a \sum_{i=1}^\infty |(B^i)_{kl}|
%\leq\sum_{\substack{k=1\\k\neq l}}^a\frac 1 {a}
=  1-\frac 1 a
%\leq \left((X_A^TX_A/n)^{\text -1}\right)_{ll}.
\end{align*}
for any $l\in A$. These three displays imply for any $l\in A$ the following inequalities:
  \begin{align*}
\sum_{\substack{k=1\\k\neq l}}^a|\left((X_A^TX_A/n)^{\text-1}\right)_{kl}|\leq 1-\frac 1 a\leq ((X_A^TX_A/n)^{\text -1})_{ll}\leq 1+\frac 1 a.
  \end{align*}
The claim can be deduce from this using $n(X_A^TX_A)^{\text -1}=(X_A^TX_A/n)^{\text -1}$.

\end{proof}

\subsection{Proof of Theorem~\ref{res:main}}

\begin{proof}[Proof of Theorem~\ref{res:main}] The first claim follows directly from Lemmata~\ref{res:lassoinvertible} and \ref{res:lsfitting}.\\
For two remaining inequalities, we invoke the KKT conditions stated in Lemma~\ref{res:kkt}; namely, we use that
\begin{align*}
\matset^T(Y-X\widehat\beta)=\frac{\lambda}{2g'\left(\|Y-X\est\|_2^2\right)}\sign(\estset)
\end{align*}
for the estimator~\eqref{eq:estimator}. The above equation can be rewritten, using the model~\eqref{eq:model}, as
% \begin{align*}
% \matset^TX\betatrue+\sigma\matset^T\epsilon-\matset^TX\est=\frac{\lambda}{g'\left(\|Y-X\est\|_2^2\right)}\sign(\estset) 
% \end{align*}
% and therefore
\begin{align*}
\matset^TX\betatrue+\sigma\matset^T\epsilon=\matset^T\matset\estset+\frac{\lambda}{2g'\left(\|Y-X\est\|_2^2\right)}\sign(\estset). 
\end{align*}
The matrix $\graset$ is invertible according to Lemma~\ref{res:lassoinvertible}, and we can therefore multiply both sides of the last equation with its inverse  $\grasetinv$. This yields
\begin{align*}
\grasetinv\matset^TX\betatrue+\sigma\grasetinv\matset^T\epsilon=\estset+\frac{\lambda}{2g'\left(\|Y-X\est\|_2^2\right)}\grasetinv\sign(\estset). 
\end{align*}
For the LS refitted estimator \eqref{eq:estimatorrefitted}, we derive similarly 
\begin{align*}
\grasetinv\matset^TX\betatrue+\sigma\grasetinv\matset^T\epsilon=\estsetref.
\end{align*}
The previous two displays imply the second equality
\begin{align*}
  \|\estref-\est\|_q&= \| \estsetref-\estset\|_q=\frac{\lambda}{2g'\left(\|Y-X\est\|_2^2\right)}\|\grasetinv\sign(\estset)\|_q.
\end{align*}
For the third claim, we find similarly
\begin{align*}
X\estref-X\est=\matset\estsetref-\matset\estset &=\frac{\lambda}{2g'\left(\|Y-X\est\|_2^2\right)}\matset\grasetinv\sign(\estset)\label{eq:kktresults}.
\end{align*}
Using the model~\eqref{eq:model} and the definition of the LS refitted estimator~\eqref{eq:estimatorrefitted}, we also find
  \begin{align*}
&\|X\betatrue-X\estref\|_2^2-    \|X\betatrue-X\est\|_2^2\\
%=&\|X\betatrue+\sigma\epsilon-X\estref-\sigma\epsilon\|_2^2-    \|X\betatrue+\sigma\epsilon-X\est-\sigma\epsilon\|_2^2\\
=&\|Y-X\estref-\sigma\epsilon\|_2^2-    \|Y-X\est-\sigma\epsilon\|_2^2\\
=&\|Y-X\estref\|_2^2+\|\sigma\epsilon\|_2^2-2\sigma<Y-X\estref,\epsilon>-\left(    \|Y-X\est\|_2^2+\|\sigma\epsilon\|_2^2-2\sigma<Y-X\est,\epsilon>\right)\\
=&\|Y-X\estref\|_2^2+2\sigma<X\estref-X\est,\epsilon>-\|Y-X\est\|_2^2\\
\leq &2\sigma<X\estref-X\est,\epsilon>.%\\
%=&\|Y-\matset\estref\|_2^2+2<\matset\estref-\matset\est,\sigma\epsilon>-\|Y-\matset\est\|_2^2.
  \end{align*}
The two previous displays yield
\begin{align*}
\|X\betatrue-X\estref\|_2^2-\|X\betatrue-&X\est\|_2^2\\
 \leq &\frac{\sigma\lambda}{g'\left(\|Y-X\est\|_2^2\right)}<\matset\grasetinv\sign(\estset),\epsilon>\\
=&\frac{\sigma\lambda}{g'\left(\|Y-X\est\|_2^2\right)}<\grasetinv\sign(\estset),\matset^T\epsilon>\\
\leq &\|\grasetinv\sign(\estset)\|_1\frac{\sigma\lambda\|\matset^T\epsilon\|_\infty}{g'\left(\|Y-X\est\|_2^2\right)}.
\end{align*}
This eventually proves the third claim.
\end{proof}

\subsection{Some technical results and remarks}\label{sec:techrm}

\begin{lemma}\label{res:convex}
  The function $\beta\mapsto g\left(\|Y-X\beta\|_2^2\right)$ is convex, and a solution of~\eqref{eq:estimator} exists.
\end{lemma}

\begin{proof}[Proof of Lemma~\ref{res:convex}]
  On may readily check that the function $\beta\mapsto \|Y-X\beta\|_2^2$ is twice continuously differentiable with Hessian $2X^TX$. The matrix $2X^TX$ is positive semidefinite, and therefore $\beta\mapsto \|Y-X\beta\|_2^2$ is convex. Consequently, 
%   \begin{align*}
%     \frac{d}{d\beta_k}\frac{d}{d\beta_l}\|Y-X\beta\|_2^2&=\frac{d}{d\beta_k}\frac{d}{d\beta_l}\sum_{i=1}^n\big(Y_i-\sum_{j=1}^pX_{ij}\beta_j\big)^2\\
%     &=-2\frac{d}{d\beta_k}\sum_{i=1}^nX_{il}\big(Y_i-\sum_{j=1}^pX_{ij}\beta_j\big)\\
%     &=2\sum_{i=1}^nX_{il}X_{ik}\\
% &=2(X^TX)_{lk}
%   \end{align*}
%for any $k,l\in\indexlasso$. The Hesse matrix of $\beta\mapsto \|Y-X\beta\|_2^2$ is thus $2X^TX$, and therefore, noting that $2X^TX$ is positive semidefinite, $\beta\mapsto \|Y-X\beta\|_2^2$ is convex. Consequently,
\begin{equation*}
\|Y-X(\alpha\beta+(1-\alpha)\gamma)\|_2^2\leq \alpha\|Y-X\beta\|_2^2+(1-\alpha)\|Y-X\gamma\|_2^2  
\end{equation*}
for any $\alpha\in[0,1]$ and $\beta,\gamma\in\rp$. This implies that  
\begin{align*}
g\left(\|Y-X(\alpha\beta+(1-\alpha)\gamma)\|_2^2\right)&\leq g\left(\alpha \|Y-X\beta\|_2^2+(1-\alpha)\|Y-X\gamma\|_2^2\right)\\
&\leq \alpha g\left(\|Y-X\beta\|_2^2\right)+(1-\alpha)g\left(\|Y-X\gamma\|_2^2\right),    
\end{align*}
since $g$ is increasing and convex. This concludes the proof of the first claim.\\
For the second claim, we note that it holds (since $g$ is increasing)
\begin{equation*}
  g\left(\|Y\|_2^2\right)\geq g\left(\|Y-X\est\|_2^2\right) + \lambda\|\est\|_1\geq g\left(0\right) + \lambda\|\est\|_1
\end{equation*}
if $\est$ is a minimizer of~\eqref{eq:estimator}. Consequently,   for finding minimal values, the function $\beta\mapsto g\left(\|Y-X\beta\|_2^2\right)+\lambda\|\beta\|_1$ can be restricted to the compact set $\mathcal B:=\left\{\beta\in\rp:\lambda\|\beta\|_1\leq g\left(\|Y\|_2^2\right)-g(0)\right\}$. The proof follows then from standard calculus.
\end{proof}

\begin{remark}[Differentiability at 0]\label{rm:sqrtlassoextension} The differentiability of $g$ in $0$ is only assumed for ease of exposition. In fact, it ensures that the KKT conditions in Lemma~\ref{res:kkt} hold true when $Y= X\est$. However, the triangle inequality and the model~\eqref{eq:model} yield
% \begin{equation*}
%   \|X\betatrue -X\est\|_2=\|X\betatrue+\sigma\epsilon -X\est-\sigma\epsilon\|_2\leq \|Y -X\est\|_2+\|\sigma\epsilon\|_2
% \end{equation*}
% and therefore
 \begin{equation*}
    \|Y-X\est\|_2\geq \sigma\|\epsilon\|_2-\|X\betatrue -X\est\|_2.
  \end{equation*}
This implies that only weak assumptions on $\|X\betatrue -X\est\|_2$ (for example, that slow rate bounds apply, cf. Remark~\ref{rm:correlations} in Section~\ref{sec:theory}) ensure $Y\neq X\est$ with high probability, which then allows one to include functions $g$ that are differentiable on $(0,\infty)$ only. 
\end{remark}

\begin{remark}[Criterion]\label{rm:heuristics}
Let us have a peek at the criterion~\eqref{eq:criterion}: Norms of the vector $\grasetinv\sign(\estset)$ play a crucial role in the theoretical results in Section~\ref{sec:theory}. However, criteria based on norms of the vector instead of the vector itself turn out to be not suitable. Indeed, the signs involved can be crucial as one can infer from the KKT conditions~\ref{res:kkt} in Section~\ref{sec:aux}. To motivate the specific form of the criterion, we consider two special cases in the following. Assuming that the correlations for the estimated active set are bounded by
\begin{equation}\label{eq:remarkheuristics2}
    |X_i^TX_j|< n/(2\setnor)~\text{~~~for all~}i,j\in \set,~ i\neq j,
\end{equation}
we can easily deduce from Lemma~\ref{res:invertibility} in Section~\ref{sec:aux} that $F(\set)=0$ as desired (cf. Remark~\ref{rm:correlations} in Section~\ref{sec:theory}).\\
In the extreme case of orthogonal designs, that is,  $X^TX=n\1$, assumption~\eqref{eq:remarkheuristics2} is naturally satisfied. Additionally, for orthogonal designs, it holds that the equation
%$\matset^T\matset=n\1$, $n\grasetinv=\1$ and $(\grasetinv\matset^TX\betatrue)_\set=\betatrue_\set.$
% \begin{align*}
%   (\grasetinv\matset^TX\betatrue)_j&= (\matset^TX\betatrue)_j/n\\
% &=  \big(\matset^T\sum_{k=1}^p\betatrue_kX_k)\big)_j/n\\
% &=  \sum_{k=1}^p\betatrue_k(\matset^TX_k)_j/n\\
% &=  \sum_{k=1}^p\betatrue_kX_j^TX_k/n\\
% &=  \betatrue_jX_j^TX_j/n\\
% &=  \betatrue_j
% \end{align*}
% for all $j\in\set$. 
%that the equation
\begin{align*}%\label{eq:remarkheuristics1}
\grasetinv\matset^TX\betatrue+&\sigma\grasetinv\matset^T\epsilon=\estset+\frac{\lambda}{2g'\left(\|Y-X\est\|_2^2\right)}\grasetinv\sign(\estset)\nonumber 
\end{align*}
derived in the proof of Theorem~\ref{res:main} simplifies to
\begin{align*}
\betatrue_{\set}+\frac{\sigma}{n}\matset^T\epsilon=\estset+\frac{\lambda}{2ng'\left(\|Y-X\est\|_2^2\right)}\sign(\estset)
\end{align*}
(we omit the straightforward derivations). Similarly, we obtain for the LS~refitted estimator 
\begin{align*}
\betatrue_\set+\frac{\sigma}{n}\matset^T\epsilon=\estsetref. 
\end{align*}
The term $\frac{\lambda}{2ng'\left(\|Y-X\est\|_2^2\right)}\sign(\estset)$ can be interpreted as a bias term for the initial estimator. It is not present for the LS refitted estimator indicating that the LS~refitted estimator is - for orthogonal designs - more accurate than the initial estimator.
\end{remark}

\subsection{Alternative implementations}\label{sec:numerics} 
For applications, it is important that the conclusions stated in Section~\ref{sec:discussion} are independent of the computational implementation. In this section, we therefore detail on our implementation and compare it to alternative ones. The results are in accord with the results of Section~\ref{sec:main} and the conclusions of Section~\ref{sec:discussion}, even though the results slightly differ among the implementations.\\

\paragraph{Computations} We considered the setting and computation of Section~\ref{sec:main} with the following details or modifications:\\
{\it glmnet}: This is the implementation applied in Section~\ref{sec:main}. For the Lasso~\eqref{eq:lasso}, we used the R-package glmnet version 1.9-3, see \cite{Hastie10,Rsoftware} for details. For the subsequent computation of the least-squares refitting~\eqref{eq:estimatorrefitted}, we used the well-known explicit formula of the least-squares estimator; however, we replaced $\graset$ by $\graset+10^{\text -7}/n\times\1$ to improve the numerical stability of the matrix inversion involved (the magnitude of the factor in front of the identity matrix~$\1$ had, as long as it was chosen in a reasonable range, no influence on the results).\\
{\it glmnet$\times$2}: We modify the above implementation by using the mentioned glmnet package for the computation of both the Lasso and the subsequent least-squares refitting (with tuning parameter $\lambda=0$ for the latter).\\
{\it stage, lar}, and {\it lasso}: For the Lasso, we used the algorithms forward stagewise, least angle regression, and lasso, respectively, of the R-package lars version 1.1, see \cite{lars,Efron-LARS} for details. The least-squares refitting was conducted as in the first implementation using the explicit formula for the least-squares estimator and replacing the matrix $\graset$ by $\graset+10^{\text -7}/n\times\1$.

\begin{table}
\begin{tabular}{l l c c c c c}
& & \bf{pred. error} & \bf est. error & \bf false neg. & \bf false pos. & \bf LS pred./est.\\
~&~&~&~&~&~\\ 
 \multirow{3}{*}{\rotatebox{90}{\mbox{glmnet}}} &{\bf Lasso} & $(5.10 \pm 0.04)\times 10^{\text -2}$ & $2.19\pm 0.02$ & $8.29\pm 0.08$ & $50.9\pm 0.4$ & 0/0 \\ 
& {\bf LS Lasso}& $(5.67 \pm 0.04)\times 10^{\text -2}$ & $3.51\pm 0.04$ & $8.29\pm 0.08$ & $50.9\pm 0.4$ & 1/1 \\
& {\bf $c$-LS Lasso}& $(5.15 \pm 0.04)\times 10^{\text -2}$ & $2.19\pm 0.02$ & $8.29\pm 0.08$ & $50.9\pm 0.4$ & 0.12/0\\
~&~&~&~&~&~\\ 
 \multirow{3}{*}{\rotatebox{90}{\mbox{glmnet$\times$2}}} &{\bf Lasso} & $(5.10 \pm 0.04)\times 10^{\text -2}$ & $2.19\pm 0.02$ & $8.29\pm 0.08$ & $50.9\pm 0.4$ & 0/0 \\ 
& {\bf LS Lasso}& $(8.16 \pm 0.07)\times 10^{\text -2}$ & $5.04\pm 0.04$ & $8.29\pm 0.08$ & $50.9\pm 0.4$ & 1/1 \\
& {\bf $c$-LS Lasso}& $(5.42 \pm 0.05)\times 10^{\text -2}$ & $2.19\pm 0.02$ & $8.29\pm 0.08$ & $50.9\pm 0.4$ & 0.115/0\\
~&~&~&~&~&~\\
\multirow{3}{*}{\rotatebox{90}{\mbox{stage}}} &{\bf Lasso} & $(6.35 \pm 0.06)\times 10^{\text -2}$ & $2.79\pm 0.08$ & $12.33\pm 0.06$ & $51.8\pm 0.4$ & 0/0 \\ 
& {\bf LS Lasso}& $(5.91 \pm 0.04)\times 10^{\text -2}$ & $3.28\pm 0.08$ & $12.33\pm 0.06$ & $51.8\pm 0.4$ & 1/1 \\
& {\bf $c$-LS Lasso}& $(6.31 \pm 0.06)\times 10^{\text -2}$ & $2.79\pm 0.08$ & $12.33\pm 0.06$ & $51.8\pm 0.4$ & 0.051/0\\
~&~&~&~&~&~\\
\multirow{3}{*}{\rotatebox{90}{\mbox{lar}}} &{\bf Lasso} & $(10.11 \pm 0.06)\times 10^{\text -2}$ & $10.10\pm 0.20$ & $10.56\pm 0.07$ & $86.8\pm 0.2$ & 0/0 \\ 
& {\bf LS Lasso}& $\textcolor{white}{0}(8.65 \pm 0.05)\times 10^{\text -2}$ & $15.79\pm 0.28$ & $10.56\pm 0.07$ & $86.8\pm 0.2$ & 1/1 \\
& {\bf $c$-LS Lasso}& $\textcolor{white}{0}(8.68 \pm 0.05)\times 10^{\text -2}$ & $11.96\pm 0.27$ & $10.56\pm 0.07$ & $86.8\pm 0.2$ & 0.96/0.221\\
~&~&~&~&~&~\\
\multirow{3}{*}{\rotatebox{90}{\mbox{lasso}}} &{\bf Lasso} & $(8.50 \pm 0.06)\times 10^{\text -2}$ & $2.50\pm 0.01$ & $11.53\pm 0.06$ & $47.9\pm 0.2$ & 0/0 \\ 
& {\bf LS Lasso}& $(4.48 \pm 0.03)\times 10^{\text -2}$ & $2.54\pm 0.01$ & $11.53\pm 0.06$ & $47.9\pm 0.2$ & 1/1 \\
& {\bf $c$-LS Lasso}& $(7.92 \pm 0.07)\times 10^{\text -2}$ & $2.50\pm 0.01$ & $11.53\pm 0.06$ & $47.9\pm 0.2$ & 0.151/0\\
~&~&~&~&~&~\\ 
& {\bf Zero}&$108.99\pm 0.01$&$2.68$&20&0&0/0\\
~&~&~&~&~&~\\
\end{tabular}
\caption{The performances of the different computational implementations for the setting with parameters 
$n=100,$ $p=1000,$ $\sigma=0.3,$ $s=20,$ and $\kappa=0.9$. The setting and the computations are described in the text.  
}\label{tab:computations}
\end{table}

\paragraph{Results and conclusions} For all sets of parameters considered in Section~\ref{sec:main}, the results of the above implementations were consistent with the results and conclusions stated in Sections~\ref{sec:main} and~\ref{sec:discussion}, even though the implementations differed from each other by their performances and computational costs.\\
For the sake of a concise presentation, we  give in Table~\ref{tab:computations} the numerical results for only one (critical) set of parameters. (We stress that the constant $c$ was kept constant throughout the calculations; slightly better performances of the $c$-LS~Lasso are expected for a more careful choice of this constant - for example, adapted to the computational implementation under consideration.)

\section*{Acknoledgments}
The author thanks Mohamed Hebiri and Sara van de Geer for their insightful comments.

\bibliography{../../Bibliography/Literature}

\begin{thebibliography}{25}
% BibTex style file: imsart-number.bst, 2010-01-14
% Default style options (sort=1,type=number).
% Used options (sort=1,type=number).

\bibitem{Belloni09}
\begin{barticle}[author]
\bauthor{\bsnm{Belloni},~\bfnm{A.}\binits{A.}} \AND
  \bauthor{\bsnm{Chernozhukov},~\bfnm{V.}\binits{V.}}
(\byear{2009}).
\btitle{Least squares after model selection in high-dimensional sparse models}.
\bjournal{preprint, http://www.arxiv.org/abs/1001.0188}.
\end{barticle}
\endbibitem

\bibitem{Belloni11}
\begin{barticle}[author]
\bauthor{\bsnm{Belloni},~\bfnm{A.}\binits{A.}},
  \bauthor{\bsnm{Chernozhukov},~\bfnm{V.}\binits{V.}} \AND
  \bauthor{\bsnm{Wang},~\bfnm{L.}\binits{L.}}
(\byear{2011}).
\btitle{Square-root lasso: pivotal recovery of sparse signals via conic
  programming}.
\bjournal{Biometrika}
\bvolume{98}
\bpages{791--806}.
\end{barticle}
\endbibitem

\bibitem{Bickel09}
\begin{barticle}[author]
\bauthor{\bsnm{Bickel},~\bfnm{P.}\binits{P.}},
  \bauthor{\bsnm{Ritov},~\bfnm{Y.}\binits{Y.}} \AND
  \bauthor{\bsnm{Tsybakov},~\bfnm{A.}\binits{A.}}
(\byear{2009}).
\btitle{Simultaneous analysis of lasso and {D}antzig selector}.
\bjournal{Ann. Statist.}
\bvolume{37}
\bpages{1705--1732}.
\end{barticle}
\endbibitem

\bibitem{Buhlmann11}
\begin{bbook}[author]
\bauthor{\bsnm{B{\"u}hlmann},~\bfnm{Peter}\binits{P.}} \AND
  \bauthor{\bparticle{van~de} \bsnm{Geer},~\bfnm{Sara}\binits{S.}}
(\byear{2011}).
\btitle{Statistics for high-dimensional data: Methods, theory and
  applications}.
\bseries{Springer Series in Statistics}.
\bpublisher{Springer Verlag}.
\bdoi{10.1007/978-3-642-20192-9}
\end{bbook}
\endbibitem

\bibitem{BuneaEN}
\begin{barticle}[author]
\bauthor{\bsnm{Bunea},~\bfnm{F.}\binits{F.}}
(\byear{2008}).
\btitle{Honest variable selection in linear and logistic regression models via
  {$\ell_1$} and {$\ell_1+\ell_2$} penalization}.
\bjournal{Electron. J. Stat.}
\bvolume{2}
\bpages{1153--1194}.
\end{barticle}
\endbibitem

\bibitem{Yoyo13}
\begin{barticle}[author]
\bauthor{\bsnm{Bunea},~\bfnm{F.}\binits{F.}},
  \bauthor{\bsnm{Lederer},~\bfnm{J.}\binits{J.}} \AND
  \bauthor{\bsnm{She},~\bfnm{Y.}\binits{Y.}}
\btitle{The Group Square-Root Lasso: Theoretical Properties and Fast
  Algorithms}.
\bjournal{preprint, http://www.arxiv.org/abs/1302.0261}.
\end{barticle}
\endbibitem

\bibitem{Candes08}
\begin{barticle}[author]
\bauthor{\bsnm{Cand{\`e}s},~\bfnm{E.}\binits{E.}},
  \bauthor{\bsnm{Wakin},~\bfnm{M.}\binits{M.}} \AND
  \bauthor{\bsnm{Boyd},~\bfnm{S.}\binits{S.}}
(\byear{2008}).
\btitle{Enhancing sparsity by reweighted {$l_1$} minimization}.
\bjournal{J. Fourier Anal. Appl.}
\bvolume{14}.
\end{barticle}
\endbibitem

\bibitem{Efron-LARS}
\begin{barticle}[author]
\bauthor{\bsnm{Efron},~\bfnm{B.}\binits{B.}},
  \bauthor{\bsnm{Hastie},~\bfnm{T.}\binits{T.}},
  \bauthor{\bsnm{Johnstone},~\bfnm{I.}\binits{I.}} \AND
  \bauthor{\bsnm{Tibshirani},~\bfnm{R.}\binits{R.}}
(\byear{2004}).
\btitle{Least angle regression}.
\bjournal{Ann. Statist.}
\bvolume{32}
\bpages{407--499}.
\bnote{With discussion, and a rejoinder by the authors}.
\end{barticle}
\endbibitem

\bibitem{Hastie10}
\begin{barticle}[author]
\bauthor{\bsnm{Friedman},~\bfnm{J.}\binits{J.}},
  \bauthor{\bsnm{Hastie},~\bfnm{T.}\binits{T.}} \AND
  \bauthor{\bsnm{Tibshirani},~\bfnm{R.}\binits{R.}}
(\byear{2010}).
\btitle{Regularization Paths for Generalized Linear Models via Coordinate
  Descent}.
\bjournal{Journal of Statistical Software}
\bvolume{33}
\bpages{1--22}.
\end{barticle}
\endbibitem

\bibitem{lars}
\begin{bmanual}[author]
\bauthor{\bsnm{Hastie},~\bfnm{Trevor}\binits{T.}} \AND
  \bauthor{\bsnm{Efron},~\bfnm{Brad}\binits{B.}}
(\byear{2012}).
\btitle{lars: Least Angle Regression, Lasso and Forward Stagewise}
\bnote{R package version 1.1, http://CRAN.R-project.org/package=lars}.
\end{bmanual}
\endbibitem

\bibitem{YoyoMomo12}
\begin{barticle}[author]
\bauthor{\bsnm{Hebiri},~\bfnm{M.}\binits{M.}} \AND
  \bauthor{\bsnm{Lederer},~\bfnm{J.}\binits{J.}}
(\byear{2013}).
\btitle{How Correlations Influence {L}asso Prediction}.
\bjournal{IEEE Trans. Inform. Theory}
\bvolume{59}
\bpages{1846-1854}.
\end{barticle}
\endbibitem

\bibitem{HCB08}
\begin{bunpublished}[author]
\bauthor{\bsnm{Huang},~\bfnm{C.}\binits{C.}},
  \bauthor{\bsnm{Cheang},~\bfnm{G.}\binits{G.}} \AND
  \bauthor{\bsnm{Barron},~\bfnm{A.}\binits{A.}}
(\byear{2008}).
\btitle{Risk of {P}enalized {L}east {S}quares, {G}reedy {S}election and {L}1
  Penalization for Flexible Function Libraries}.
\bnote{Manuscript}.
\end{bunpublished}
\endbibitem

\bibitem{Kolt10}
\begin{barticle}[author]
\bauthor{\bsnm{Koltchinskii},~\bfnm{V.}\binits{V.}},
  \bauthor{\bsnm{Lounici},~\bfnm{K.}\binits{K.}} \AND
  \bauthor{\bsnm{Tsybakov},~\bfnm{A.}\binits{A.}}
(\byear{2011}).
\btitle{Nuclear-norm penalization and optimal rates for noisy low-rank matrix
  completion}.
\bjournal{Ann. Statist.}
\bvolume{39}
\bpages{2302--2329}.
\end{barticle}
\endbibitem

\bibitem{kkt}
\begin{binproceedings}[author]
\bauthor{\bsnm{Kuhn},~\bfnm{H.}\binits{H.}} \AND
  \bauthor{\bsnm{Tucker},~\bfnm{A.}\binits{A.}}
(\byear{1951}).
\btitle{Nonlinear programming}.
In \bbooktitle{Proceedings of the second Berkeley symposium on mathematical
  statistics and probability}
\bvolume{5}
\bpages{481--492}.
\end{binproceedings}
\endbibitem

\bibitem{Karim08}
\begin{barticle}[author]
\bauthor{\bsnm{Lounici},~\bfnm{K.}\binits{K.}}
(\byear{2008}).
\btitle{Sup-norm convergence rate and sign concentration property of {L}asso
  and {D}antzig estimators}.
\bjournal{Electron. J. Stat.}
\bvolume{2}
\bpages{90--102}.
\end{barticle}
\endbibitem

\bibitem{MM11}
\begin{barticle}[author]
\bauthor{\bsnm{Massart},~\bfnm{P.}\binits{P.}} \AND
  \bauthor{\bsnm{Meynet},~\bfnm{C.}\binits{C.}}
(\byear{2011}).
\btitle{The {L}asso as an $\ell_1$-ball model selection procedure}.
\bjournal{Electron. J. Stat.}
\bvolume{5}
\bpages{669--687}.
\end{barticle}
\endbibitem

\bibitem{Meinsthresh09}
\begin{barticle}[author]
\bauthor{\bsnm{Meinshausen},~\bfnm{N.}\binits{N.}} \AND
  \bauthor{\bsnm{Yu},~\bfnm{B.}\binits{B.}}
(\byear{2009}).
\btitle{Lasso-type recovery of sparse representations for high-dimensional
  data}.
\bjournal{Ann. Statist.}
\bvolume{37}
\bpages{246--270}.
\end{barticle}
\endbibitem

\bibitem{Rsoftware}
\begin{bmanual}[author]
\bauthor{\bsnm{{R Core Team}}}
(\byear{2013}).
\btitle{R: A Language and Environment for Statistical Computing}, \baddress{R
  Foundation for Statistical Computing, Vienna, Austria}
\bnote{http://www.R-project.org/}.
\end{bmanual}
\endbibitem

\bibitem{RigTsy11}
\begin{barticle}[author]
\bauthor{\bsnm{Rigollet},~\bfnm{P.}\binits{P.}} \AND
  \bauthor{\bsnm{Tsybakov},~\bfnm{A.}\binits{A.}}
(\byear{2011}).
\btitle{{E}xponential {S}creening and optimal rates of sparse estimation}.
\bjournal{Ann. Statist.}
\bvolume{39}
\bpages{731--771}.
\end{barticle}
\endbibitem

\bibitem{ScaledLasso11}
\begin{barticle}[author]
\bauthor{\bsnm{Sun},~\bfnm{T.}\binits{T.}} \AND
  \bauthor{\bsnm{Zhang},~\bfnm{{C. -H. }}\binits{C.}}
(\byear{2012}).
\btitle{Scaled sparse linear regression}.
\bjournal{Biometrika}
\bvolume{99}
\bpages{879--898}.
\end{barticle}
\endbibitem

\bibitem{Tibshirani-LASSO}
\begin{barticle}[author]
\bauthor{\bsnm{Tibshirani},~\bfnm{R.}\binits{R.}}
(\byear{1996}).
\btitle{Regression shrinkage and selection via the lasso}.
\bjournal{J. Roy. Statist. Soc. Ser. B}
\bvolume{58}
\bpages{267--288}.
\end{barticle}
\endbibitem

\bibitem{Sara09}
\begin{barticle}[author]
\bauthor{\bparticle{van~de} \bsnm{Geer},~\bfnm{S.}\binits{S.}} \AND
  \bauthor{\bsnm{B{\"u}hlmann},~\bfnm{P.}\binits{P.}}
(\byear{2009}).
\btitle{On the conditions used to prove oracle results for the {L}asso}.
\bjournal{Electron. J. Stat.}
\bvolume{3}
\bpages{1360--1392}.
\end{barticle}
\endbibitem

\bibitem{SaraShu11}
\begin{barticle}[author]
\bauthor{\bparticle{van~de} \bsnm{Geer},~\bfnm{Sara}\binits{S.}},
  \bauthor{\bsnm{B{\"u}hlmann},~\bfnm{Peter}\binits{P.}} \AND
  \bauthor{\bsnm{Zhou},~\bfnm{Shuheng}\binits{S.}}
(\byear{2011}).
\btitle{The adaptive and the thresholded {L}asso for potentially misspecified
  models (and a lower bound for the {L}asso)}.
\bjournal{Electron. J. Stat.}
\bvolume{5}
\bpages{688--749}.
\end{barticle}
\endbibitem

\bibitem{vdGeer11}
\begin{barticle}[author]
\bauthor{\bparticle{van~de} \bsnm{Geer},~\bfnm{S.}\binits{S.}} \AND
  \bauthor{\bsnm{Lederer},~\bfnm{J.}\binits{J.}}
(\byear{2013}).
\btitle{The {L}asso, correlated design, and improved oracle inequalities}.
\bjournal{IMS Collections}
\bvolume{9}
\bpages{303-316}.
\end{barticle}
\endbibitem

\bibitem{BiYuConsistLasso}
\begin{barticle}[author]
\bauthor{\bsnm{Zhao},~\bfnm{P.}\binits{P.}} \AND
  \bauthor{\bsnm{Yu},~\bfnm{B.}\binits{B.}}
(\byear{2006}).
\btitle{On model selection consistency of {L}asso}.
\bjournal{J. Mach. Learn. Res.}
\bvolume{7}
\bpages{2541--2563}.
\end{barticle}
\endbibitem

\end{thebibliography}
%\bibliography{Literature}

\end{document}